\documentclass[11pt,reqno,a4paper]{amsart}

\usepackage[english,activeacute]{babel}
\usepackage{enumerate,xspace,upref}
\usepackage{graphicx}
\usepackage{stmaryrd}
\DeclareMathOperator*{\bigtimes}{\vartimes}
\usepackage{hyperref}
	\hypersetup{citebordercolor=.1 .6 1}

\IfFileExists{mmymtpro2.sty}{\usepackage[mtpccal,mtpscr,subscriptcorrection]{mymtpro2}
	\let\ssave\mathcal
  \let\mathcal\mathscr
  \let\mathscr\ssave

  \def\cup{\cupprod}
  \def\cap{\capprod}
  \def\bigcup{\bigcupprod}
  \def\bigcap{\bigcapprod}
  \def\bigcupdisjoint{\mathop{\kern10pt\raisebox{4pt}{$\cdot$}\kern-12pt\bigcup}\limits}
	\usepackage[notref,notcite]{showkeys}}%
	{\usepackage{amsmath,amsthm,amssymb}
	 \usepackage{epsfig}
	\usepackage{latexsym}
	\usepackage{pstricks,pst-grad}
	\usepackage[active]{srcltx}
}

\numberwithin{equation}{section}
\newtheorem{thm}{Theorem}[section]
\newtheorem{cor}[thm]{Corollary}
\newtheorem{lem}[thm]{Lemma}
\newtheorem{prop}[thm]{Proposition}

\theoremstyle{definition}
\newtheorem{defn}[thm]{Definition}

\theoremstyle{remark}
\newtheorem{rem}[thm]{Remark}

\newcommand{\be}{\begin{equation}}
\newcommand{\ee}{\end{equation}}
\newcommand{\ba}{\begin{array}}
\newcommand{\ea}{\end{array}}
\newcommand{\bal}{\begin{align}}
\newcommand{\eal}{\end{align}}
\newcommand{\bea}{\begin{eqnarray}}
\newcommand{\eea}{\end{eqnarray}}
\newcommand{\bee}{\begin{eqnarray*}}
\newcommand{\eee}{\end{eqnarray*}}
\newcommand{\tr}{{\textrm{tr}}}
\renewcommand{\P}{\mathbb P}
\newcommand{\Z}{\mathbb Z^d}
\newcommand{\Rd}{\mathbb R^d}
\newcommand{\R}{\mathbb R}
\newcommand{\C}{\mathbb C}
\newcommand{\N}{\mathbb N}
\renewcommand{\AA}{\mathbb A}

\newcommand{\cC}{\mathcal C}
\newcommand{\cP}{\mathcal P}

\newcommand{\Az}{{\upshape\textbf{(A0)}}\xspace}
\newcommand{\Ao}{{\upshape\textbf{(A1)}}\xspace}

\newcommand{\At}{{\upshape\textbf{(A2)}}\xspace}
\newcommand{\Aloc}{{\upshape\textbf{(A3)}}\xspace}
\newcommand{\Alt}{{\upshape\textbf{(A4)}}\xspace}

\let\savechi\chi
\def\chi{\raisebox{.3ex}{$\savechi$}}

\DeclareMathOperator{\sydi}{\raisebox{.6pt}{\scriptsize$\bigtriangleup$}}
\newcommand{\diam}{\mathop{\mathrm{diam}}}

\newcommand{\dist}{\mathop{\mathrm{dist}}}
\newcommand{\spec}{\mathop{\mathrm{spec}}}
\newcommand{\supp}{{\mathop{\mathrm{supp\,}}}}
\newcommand{\Tr}{{\mathop{\mathrm{tr} \,}}}
\DeclareMathOperator{\e}{e}
\renewcommand{\d}{\mathrm{d}}
\renewcommand{\i}{\mathrm{i}}

\newcommand{\Prob}[1]{\mathbb P\left(#1\right)}
\newcommand{\norm}[1]{\Vert #1 \Vert}
\newcommand{\abs}[1]{\left| #1 \right|}

\newcommand{\Lp}[1]{\textrm{L}^2(#1)}
\renewcommand{\L}{\Lambda}

\newcommand{\angles}[1]{\langle #1 \rangle}

\newcommand{\pom}{{P^\omega}}

\newcommand{\E}{\mathbb E}

\newcommand{\hdom}{H_{D^\omega}}

\newcommand{\dom}{{D^\omega}}
\newcommand{\om}{\omega}

\newcommand{\tpom}{{P^{\tilde\omega}}}

\newcommand{\hm}[1]{\leavevmode{\marginpar{\tiny%
$\hbox to 0mm{\hspace*{-0.5mm}$\leftarrow$\hss}%
\vcenter{\vrule depth 0.1mm height 0.1mm width \the\marginparwidth}%
\hbox to
0mm{\hss$\rightarrow$\hspace*{-0.5mm}}$\\\relax\raggedright #1}}}

\newtheorem{myrems}[thm]{Remarks}
\newenvironment{Remarks}{\begin{myrems}\begin{nummer}}%
    {\end{nummer}\end{myrems}}

\newcounter{numcount}
\newcommand{\labelnummer}{(\roman{numcount})}%

\makeatletter
\providecommand{\showkeyslabelformat}[1]{\relax}        
\let\mysaveformat\showkeyslabelformat                   %
\def\myformat#1{\raisebox{-1.5ex}{\mysaveformat{#1}}}   %

\newenvironment{nummer}%
  {\let\curlabelspeicher\@currentlabel%
    \begin{list}{\textup{\labelnummer}}%
      {\usecounter{numcount}\leftmargin0pt%
        \topsep0.5ex\partopsep2ex\parsep0pt\itemsep0ex\@plus1\p@%
        \labelwidth2.5em\itemindent3.5em\labelsep1em%
      }%
    \let\saveitem\item%
    \def\item{\saveitem%
      \def\@currentlabel{\curlabelspeicher\kern.1em\labelnummer}}%
    \let\savelabel\label%
    \def\label##1{{\ifnum\thenumcount=1\let\showkeyslabelformat\myformat\fi\savelabel{##1}}%
										{\def\@currentlabel{\labelnummer}%
									 	\let\showkeyslabelformat\@gobble
									 	\savelabel{##1item}%
										}%
	   							}%
  }{\end{list}}%

\newenvironment{indentnummer}%
  {\let\curlabelspeicher\@currentlabel%
    \begin{list}{\textup{\labelnummer}}%
      {\usecounter{numcount}\leftmargin0pt%
        \topsep0.5ex\partopsep2ex\parsep0pt\itemsep0ex\@plus1\p@%
        \labelwidth2.5em\itemindent0em\labelsep1em%
        \leftmargin2.5em}%
    \let\saveitem\item%
    \def\item{\saveitem%
      \def\@currentlabel{\curlabelspeicher\kern.1em\labelnummer}}%
    \let\savelabel\label%
    \def\label##1{{\ifnum\thenumcount=1\let\showkeyslabelformat\myformat\fi\savelabel{##1}}%
										{\def\@currentlabel{\labelnummer}%
									 	\let\showkeyslabelformat\@gobble
									 	\savelabel{##1item}%
										}%
    							}%
  }{\end{list}}%

\def\itemref#1{\ref{#1item}}

\begin{document}

\title[Delone--Anderson Operators]{Ergodicity and dynamical localization for Delone--Anderson operators}
\author{Fran\c{c}ois Germinet}
\address{Universit\'e de Cergy-Pontoise, UMR CNRS 8088, F-95000 Cergy-Pontoise, France}
\email{francois.germinet@u-cergy.fr}

\author{Peter M\"uller}
\address{Mathematisches Institut, Ludwig-Maximilians-Universit\"at, Theresienstr.\ 39,
	80333 M\"unchen, Germany}
\email{mueller@lmu.de}

\author{Constanza Rojas-Molina}
\address{Mathematisches Institut, Ludwig-Maximilians-Universit\"at, Theresienstr.\ 39,
	80333 M\"unchen, Germany}
\email{crojasm@math.lmu.de}

\keywords{random Schr\"odinger operators, Delone sets, Delone-Anderson operators,
	integrated density of states, ergodic theorem}


\begin{abstract}
We study the ergodic properties of Delone-Anderson operators, using the framework of randomly coloured Delone sets and Delone dynamical systems. In particular, we show the existence of the integrated density of states
and, under some assumptions on the geometric complexity of the underlying Delone sets, we obtain information on the almost-sure spectrum of the family of random operators.
We then exploit these results to study the Lifshitz-tail behaviour of the integrated density of states of a Delone--Anderson operator at the bottom of the spectrum. Furthermore, we use Lifshitz-tail estimates as an input for the multi-scale analysis to prove dynamical localization. 
\end{abstract}

\maketitle

\section{Introduction}\label{intro}
For more than 50 years, the Anderson model has been the subject of extensive studies in the mathematics and physics literature to illuminate electronic transport properties in disordered media \cite{And58,GMP77,FS83,Kirsch89,CaLa90,PF92,AM93,CH,KSS98,GK1,FLM00,S,Ue,BoK,AENSS06,GHK07,K,GK11}. In this paper, we study Anderson-type operators that are relevant for \emph{disordered aperiodic media}. The model is a variant of the well-known continuum (or alloy-type) Anderson model in that the impurities are not located at the points of the periodic hypercubic lattice but on a rather general point set. As usual, each impurity gives rise to the same single-site potential, except for a random coupling constant which mimics the various species of atoms that make up the material. Besides having constitutionally disordered aperiodic media in mind, our study is also motivated by the quest for universality in alloy-type Anderson models: details of the impurities' positions, should not affect the model's key properties.


In order to describe our model in detail we introduce some notation. Let
$\Lambda_{L}(x) := \bigtimes_{j=1}^{d} ]x_{j} - L/2, x_{j}+ L/2[$ be the open cube in $\Rd$ with edges of length $L>0$ centred at $x=(x_{1},\ldots,x_{d})\in\R^{d}$ (and oriented parallel to the coordinate axes). If the cube is centred about the origin, we simply write $\Lambda_{L} := \Lambda_{L}(0)$.

\begin{defn}\label{def_delone}
A subset $D$ of $\Rd$ is called an $(r,R)$-\emph{Delone set}
if ~(i)~~ it is \emph{uniformly discrete}, i.e.\ there exists a real $r>0$ such that $\big|D\cap\L_r(x)\big|\leq 1$  for every $x\in\Rd$, and ~(ii)~~ it is \emph{relatively dense}, i.e.\ there exists a real $R\geq r$ such that $\big|D\cap\L_R(x)\big|\geq 1$ for every $x\in\Rd$. Here, $|\cdot|$ stands for cardinality of a set.
\end{defn}

Clearly, the minimal distance between any two points in an $(r,R)$-Delone set is $r$. Also, given any point in an $(r,R)$-Delone set, one can find another point that is no further than $\sqrt{d}R$ apart.
Particular examples of Delone sets are the hypercubic lattice $\mathbb{Z}^{d}$, the vertices of a Penrose tiling or the random point set obtained from removing every other point of $\Z$ by a Bernoulli percolation process. Generally speaking, Delone sets cover a wide range from perfectly ordered point sets to strongly disordered ones.

For the rest of this paper we fix $0 < r \le R<\infty$. Given an $(r,R)$-Delone set $D$ in $\Rd$, we consider the random Schr\"odinger operator
\be\label{hdom} H_{\dom}:= H_0+V_{\dom} \ee
with dense domain in the Hilbert space $\Lp{\Rd}$ and subject to the following assumptions. Our notation $H_\dom$ for the dependence of the operator on the Delone set and on the random  coupling constants will be justified in Sect.~\ref{colouring}.

\begin{itemize}
\item[\Az]
	The background operator is either the negative Laplacian $H_0 :=-\Delta$ or, if $d=2$,
	we also allow for the Landau Hamiltonian $H_0 := (-\i\nabla-{A})^2$ with constant magnetic field $B \ge 0$
	and vector potential $\R^{2} \ni x \mapsto A(x):=\frac{B}{2}(x_2,-x_1)$ in the symmetric gauge.
	The random potential is given by
	\be \label{ranpot1}
		\R^{d} \ni x  \mapsto V_{\dom}(x) := \sum_{p \in D}\omega_p u(x-p)
	\ee
	with a compactly supported \emph{single-site potential} $u \in \mathrm{L}^{\infty}(\Rd)$.
\item[(\textbf{A1})]
	The (canonically realized) random coupling constants
	$\omega:=(\omega_{p})_{p\in D}$ are independently and identically distributed, each according to the same Borel probability measure $\P^{(0)}$ with compact support $\mathbb{A} \subset\R$, that is,
	\begin{equation} \label{prod-meas}
 	\P_D:= \bigotimes_{p\in D}\mathbb \P^{(0)}
\end{equation}
is defined on the product Borel-$\sigma$-algebra of the probability space $\Omega_{D}:=\bigtimes_{D}\R$. We denote by $\E_{D}$ the expectation with respect to $\P_D$ and define $w:=\sup_{v\in \mathbb A} \abs{v}$.

	
\end{itemize}
We infer from \Az and uniform discreteness that there exists a constant $v_{0} \in ]0,\infty[$, which depends on $u$ and $r$ (but not on the particular $D$), such that $\big\| \sum_{p \in D} |u(\;\pmb\cdot\;-p)| \big\|_{\infty} \le v_{0}$. Hence, we have
\begin{equation} \label{V-bound}
 	\| V_{\dom} \|_{\infty} \le w v_{0}
\end{equation}
for $\P_{D}$-a.e.\ $\omega\in\Omega_{D}$. Moreover, the map
$\Omega_{D} \times\Rd \ni (\omega,x) \mapsto V_{\dom}(x)$ is measurable. Therefore, the map $\Omega_{D} \ni \omega\mapsto H_{\dom}$ is also measurable. We refer to it as the \emph{Delone--Anderson operator}.

In addition to the hypotheses \textbf{(A0)} and (\textbf{A1}), the following one will be assumed in some parts of this paper.
\begin{itemize}

\item[(\textbf{A2})]
	The single-site potential is continuously differentiable with compact
	support, $u \in C_{c}^{1}(\Rd)$.
\end{itemize}

The particular case $D=\Z$ defines the usual alloy-type Anderson model, if \Az and \Ao are assumed. A fundamental consequence of the periodicity of $D$ and of the i.i.d.\ distribution of the $(\omega_{p})_{p\in D}$ is \emph{ergodicity}. Namely, there exist measure-preserving ergodic transformations $\{\tau_{a}\}_{a\in\Z}$ on $\Omega_{\Z}$ and a family of unitary (magnetic) translation operators $\{U_a\}_{a\in\Z}$ acting on $\mathrm{L}^2(\Rd)$ such that
\be\label{ergo}   H_{(\Z)^{\tau_{a}(\omega)}}= U_a H_{(\Z)^{\omega}} U_a^*\ee
for every $a\in\Z$.
Several groundbreaking studies of the Anderson model concerned spectral properties that are consequences of ergodicity. We mention the self-averaging of the integrated density of states  and the almost-surely non-random spectrum of the random family $\{H_{(\Z)^{\omega}}\}_{\omega\in\Omega_{\Z}}$, a property which also extends to each spectral component in the Lebesgue decomposition \cite{Pas80,KuSo80,KMa}.
As a consequence, studies of the spectral type of $H_{(\Z)^{\omega}}$ or properties of the dynamics generated by $H_{(\Z)^{\omega}}$ in a certain energy interval are well-defined problems that do not depend upon the chosen realization $\omega$ of coupling constants with probability one.
In the case of a general Delone set $D$ instead of $\Z$, the Delone--Anderson operator does not satisfy the particular covariance relation \eqref{ergo}, because the Delone set and its translate $a+D$ will not agree, see
\eqref{delone-trans} below instead. Thus, all of the above-mentioned consequences of ergodicity do not necessarily apply any more.

Delone--Anderson operators have been studied in the literature before.
Almost exclusively, the focus has been on proving dynamical localization over the last years, using both the fractional moment method \cite{BdMNSS} and the multi-scale analysis \cite{RM}. In the latter approach, it was shown that the bootstrap multi-scale analysis (MSA) from \cite{GK1}, and therefore, the phenomenon of dynamical localization, is insensitive to perturbations of the underlying periodic arrangement of impurities whenever this arrangement does not exhibit arbitrarily large holes. In  \cite{G}, the case $H_0=-\Delta$ was studied, using the MSA by Bourgain-Kenig \cite{BoK}. To consider more general unperturbed operators with aperiodic structures one needs unique continuation principles, which  were obtained in \cite{RMV}, together with Wegner estimates. Klein later on improved these Wegner estimates and proved dynamical localization at high disorder using the MSA method \cite{Kl}. A more involved treatment was needed in the discrete setting \cite{EKl}, where unique continuation principles are not
available. There, dynamical localization was shown at low energies, with a proof that extends to the continuous setting. A simpler approach was given by \cite{RM13}, using a space-averaging approximation as in \cite{BoK, G}.

Up to now, information on the almost-sure spectrum of $H_\dom$ has been obtained for particular cases only. If $H_0=-\Delta$ and $V_{\dom} \ge0 $ almost surely such that zero belongs to the support of the single-site distribution, then $\sigma(H_{0}) = \sigma(H_{\dom})$ almost surely, which follows from a Borel--Cantelli argument. This allows to conclude dynamical localization at the bottom of the spectrum from the MSA with probability one \cite{BdMNSS,G}.
In the case of the Landau Hamiltonian, the same argument gives the inclusion
\be  \sigma(H_0)\subset \sigma(H_\dom) \quad\mbox{almost surely.}\ee
This tells only that the Landau levels are contained in the spectrum, but provides no information on the location of the band edges. To take care of this, an extra argument based on \cite{CH} was needed in \cite{RM} to show that the intersection between the region of dynamical localization and the spectrum of the realizations $H_{\dom}$ is not empty almost surely -- but its location can depend, possibly, on the realization.

The purpose of this paper is twofold. In the first part, Section~\ref{s:erg}, we introduce a dynamical system for randomly coloured Delone sets and study the ergodicity properties of Delone--Anderson operators within this framework. In particular, we prove the existence and self-averaging of the integrated density of states of Delone--Anderson operators in Corollary~\ref{eids}.
In Theorem~\ref{t:supp-ids} we obtain a description of the almost-sure spectrum of $H_\dom$ in terms of the growth points of the integrated density of states.
To our knowledge, this is the first study of ergodic properties for the Delone--Anderson model.
In combination with the results from \cite{RM}, this allows to conclude band-edge localization (in the usual sense) for the Landau model with a Delone--Anderson potential with probability one.

The dynamical system for Delone--Anderson operators builds upon the Delone \emph{hull} $X_D=\overline{\{x+D:\, x\in\Rd\}}$ of a given Delone set $D$, which is a suitably defined closure of the set of all translates of $D$, see Definition~\ref{hull-def} below.
Existence of the integrated density of states and non-randomness of the spectrum
hold for almost-every (w.r.t.\ a suitable measure) point set in the Delone hull. If the Delone hull of $D$ satisfies stronger hypotheses, then
these properties hold
even for \emph{all} point sets in $X_D$. In particular,
they hold
for the given Delone set $D$, see Corollary \ref{eids}(ii) and Theorem~\ref{t:supp-ids}(ii).
In the Appendix we give an example of a non-uniquely ergodic Delone set. There we show that unique ergodicity is an essential assumption for Corollary~\ref{eids} to hold without exceptional Delone sets.

In the second part of this article, Section~\ref{s:LTdynloc}, we analyze the non-magnetic case $H_0=-\Delta$ and show that the integrated density of states exhibits a Lifshitz tail at the bottom of the spectrum in Theorem~\ref{lt}. This argument requires a suitable version of Dirichlet--Neumann bracketing averaged over the hull. In contrary to the usual proof of Lifshitz tails for the Anderson model, we need  additional efforts for making the Lifshitz-tail bounds for the finite-volume integrated density of states useful to prove localization for \emph{every} Delone set in the hull $X_{D}$, in particular for $D$ itself. Usually, the length $L \sim E^{-1/2}$ is determined and fixed by the energy in such estimates. We need to extend them to all sufficiently large lengths $L \gtrsim E^{-1/2}$. It is also important that the constants in the estimates are uniform on the hull $X_{D}$. We remark that we do not know monotonicity in $L$ of the finite-volume Dirichlet or Neumann integrated densities of states, which 
holds pointwise for every Delone set in the hull $X_{D}$ -- we know it only in average over the hull.



In Theorem \ref{t:dynlocas} we use the estimates involved in the proof of Lifshitz tails from Theorem~\ref{unifbd} to establish the initial estimate of the MSA -- this is a new way of proving the initial estimate for Delone--Anderson operators, and we show that this is possible for every Delone set in the hull $X_{D}$, even for those exceptional Delone sets for which the finite-volume integrated density of states is not known to converge in the macroscopic limit. 
The advantage of our approach concerns the size of the region of dynamical localization.
In Theorems \ref{t:dynlocas} and \ref{t:dynloc} we obtain lower bounds on the size of the interval of dynamical localization for Delone--Anderson operators within our approach and previous ones. These bounds depend on the Delone set $D$ only through its radius $R$ of relative denseness. The advantage of using the Lifshitz-tail estimates from Theorem~\ref{unifbd} is that it gives a better lower bound.

\smallskip
\section{Ergodicity and the integrated density of states}\label{s:erg}

The lack of translation covariance of Delone--Anderson operators poses the question whether they possess a self-averaging integrated density of states. For ``sufficiently regular'' Delone sets we will obtain a positive answer. Technically, we use an ergodic theorem for randomly coloured point sets developed in \cite{MR}. The next section recalls the  the framework of dynamical systems for randomly coloured point sets which we need for this purpose.

\subsection{Randomly coloured point sets}\label{colouring}
In this section we introduce the basic notions to formulate a version of the ergodic theorem for randomly coloured point sets from \cite{MR}. We point the reader to the literature in \cite{MR} for a discussion of earlier references on this topic.

Our setting of Delone sets in $d$-dimensional Euclidean space corresponds to the choices $M=\Rd$ and $T=\Rd$ (both equipped with the Euclidean topology) as point space and transformation group in \cite{MR}, respectively. The Abelian and unimodular group $\Rd$ acts on the point space $\Rd$ via translations, $\Rd \times \Rd \ni (x,y) \mapsto x+y \in \Rd$. This action is continuous, proper, transitive and free. In particular, all hypotheses required in \cite{MR} are fulfilled, see \cite[Ex.~2.5]{MR}.

Now we consider the space
\begin{equation}
 \cP_{r}(\Rd):= \big\{ P \subset \Rd: \text{$P$ is uniformly discrete of radius $r>0$}\big\}
\end{equation}
of uniformly discrete point sets in $\Rd$.  Most importantly, one can prove by standard arguments, see e.g.\  \cite[Prop.~2.9]{MR}, that $\cP_{r}(\Rd)$  is compact with respect to the \emph{vague toplogy}. This is the coarsest topology such that for every function
$\varphi \in C_{c}(\Rd)$ the map $\cP_{r}(\Rd) \ni P \mapsto \sum_{p\in P} \varphi(p)$ is continuous.
We refer to Remark~2.7 and Lemma~2.8 in \cite{MR} for different characterisations of the vague topology.

The group $\Rd$ induces a natural translation action on $\cP_{r}(\Rd)$ by setting
\begin{equation}
 	x+P := \{x+ p: p \in P\}
\end{equation}
for $x\in \Rd$ and $P\in\cP_{r}(\Rd)$. Not surprisingly, this action can be shown to be continuous.

\begin{defn}
\label{hull-def}
 For $D\subset \Rd$ an $(r,R)$-Delone set, we define its closed $\Rd$-\emph{orbit}
\be X_D:=\overline{ \{ x+ D:\, x\in \Rd \} } \subset \mathcal P_r(\Rd)\ee
where the closure is taken with respect to the vague topology. In particular, the orbit $X_{D}$ is compact in the vague topology. The triple consisting of $X_{D}$, the group $\Rd$ and its continuous translation action on $X_D$ constitutes a compact dynamical system.
\end{defn}

Next, we consider a Borel-measurable subset $\mathbb A \subseteq \R$, which we refer to as the \emph{colour space}. For a uniformly discrete point set $P \in\cP_{r}(\Rd)$ we  introduce the product space
\be \Omega_P:= \bigotimes\nolimits_{p\in P}\mathbb A \ee
of its possible colour realisations $\omega \equiv (\omega(p))_{p\in P}\in\Omega_{P}$, where $\omega(p)
\in \AA$ for all $p\in P$.
This gives rise to the coloured point set
\be P^\omega:=\{ (p,\omega(p)): \,p\in P, \omega \in \Omega_{P} \}\ee
and the coloured orbit
\be \label{col-orbit} \hat{X}_{D} := \big\{ P^{\omega}: P \in X_{D}, \omega\in \Omega_{P}\big\}
=\overline{\{ x+ D^\omega: \, x\in \Rd, \omega\in \Omega_{D} \} }\ee
of a Delone set $D$.
The equality in \eqref{col-orbit} is proved in \cite[Lemma~3.6]{MR}. In the right expression of \eqref{col-orbit} we introduced the translation of a coloured point set
$x+ P^{\omega} := (x+P)^{\tau_{x}(\omega)}$,
where $x\in\Rd$, $P\in\cP_{r}(\Rd)$ and
\be \label{taudef}
	\tau_{x} : \; \begin{array}{c@{\quad}c@{\quad}c} \Omega_{P} & \rightarrow & \Omega_{x+P} \\
		\big(\omega(p)\big)_{p\in P} & \mapsto & \big(\omega(p)\big)_{x+p\in x+ P} \end{array}.
\ee
This means that the colour is simply translated along with each point of $P$.
Furthermore, the closure in the right expression of \eqref{col-orbit} is taken with respect to the vague topology on the space of uniformly discrete coloured point sets
\begin{equation}
 	\cC_{r}(\Rd) := \{P^{\omega}: P\in\cP_{r}(\Rd), \omega\in\Omega_{P}\}.
\end{equation}
This is the coarsest topology such that for every function
$\varphi \in C_{c}(\Rd\times \AA)$ the map $\cC_{r}(\Rd) \ni P^{\omega} \mapsto \sum_{p\in P} \varphi(p,\omega(p))$ is continuous.

The above defined translation of a coloured point set is a continuous map on the compact space $\cC_{r}(\Rd)$ \cite[Lemma~3.6]{MR}, and the coloured orbit $\hat X_{D}$ is itself compact in the vague topology for every $(r,R)$-Delone set $D \subset\Rd$ \cite[Prop.~3.5]{MR}.
We note that the triple consisting of $\hat X_D$,  the group $\Rd$ and its continuous translation action on $\hat X_D$ is a compact topological dynamical system.

If assumption \Ao holds, then colours are distributed independently and identically at every point of $D$.
According to \cite[Lemma 3.9~(i)]{MR}, the product measure $\P_{D}$ from \eqref{prod-meas} satisfies all hypotheses needed in the ergodic theorem for randomly coloured point sets in \cite[Thm.\ 3.11]{MR}, but we could have also allowed for more general probability measures.

Before we state the version of the ergodic theorem \cite[Thm.\ 3.11]{MR} that we need in our setting, we remark that on every compact topological dynamical system there exists an ergodic Borel probability measure, cf.\ \cite[\S 6.2]{Wal82}. If this measure is unique, then the dynamical system is called \emph{uniquely ergodic.}

\begin{thm}\label{ergthm}
Let $D\subset \R$ be a Delone set, let $\mu$ be an ergodic Borel probability measure on $X_D$ and assume \Ao. Then there exists an ergodic probability measure $\hat \mu$ on $\hat X_D$, which is uniquely determined by $\mu$, such that the following holds.
\begin{indentnummer}
 \item For every $\Phi\in \textrm{L}^1(\hat X_D,\hat\mu)$ we have
\be\label{erg1} \int_{\hat X_D}\Phi(\tilde{P}^{\tilde{\omega}})\, \d\hat\mu(\tilde{P}^{\tilde{\omega}})= \int_{X_D} \left(\int_{\Omega_{\tilde{P}}}\Phi(\tilde{P}^{\tilde{\omega}})\, \d\P_{\tilde{P}}(\tilde\omega) \right) \d\mu(\tilde P). \ee
\item For every $\Phi\in \textrm{L}^1(\hat X_D,\hat\mu)$
the limit
\be \label{erg2} \lim_{L\rightarrow \infty}\frac{1}{L^{d}}\int_{\L_{L}}\Phi(x+ P^{\omega}) \,\d x
= \int_{\hat X_D} \Phi(\tilde{P}^{\tilde{\omega}}) \, \d\hat\mu(\tilde{P}^{\tilde{\omega}}) \ee
exists for $\hat\mu$-a.e.\ $P^{\omega}\in\hat X_D$.
Moreover, if $X_D$ is even uniquely ergodic and if $\Phi$ is continuous, then the limit \eqref{erg2} exists \emph{for every} $P\in X_D$ and for $\P_{P}$-a.a.\ $\omega\in\Omega_{P}$.
\end{indentnummer}
\end{thm}

\begin{Remarks}
\item The ergodic theorem will be most useful in the uniquely ergodic situation
  and for continuous $\Phi$. In this case the limit in \eqref{erg2} exists for $P=D$,
  the Delone set we started with, and for $\P_{D}$-a.e.\ $\omega\in\Omega_D$.
\item Such a type of ergodic theorem appeared first in the literature in \cite{Hof98}
	in the context of percolation on the Penrose tiling. While restricted to dynamical systems
	of finite local complexity (see e.g.\ \cite{MR} for a definition), the approach of \cite{Hof98}
	provided already the optimal treatment of exceptional sets for uniquely ergodic
	systems and continuous functions. A few years ago, the restriction to finite local complexity
	could be dispensed with in \cite[Lemma 10]{Len08} -- but without providing an optimal treatment
	of exceptional sets. Theorem~\ref{ergthm} is a special case of \cite[Thm.\ 3.11]{MR},
	which unites the benefits of the aforementioned approaches.
\item	\label{on-unique}
	A sufficient condition for unique ergodicity of $X_{D}$ is \emph{almost linear repetitivity} of $D$
	\cite[Prop.\ 4.4]{FrRi12}.
	In the more special case of Delone sets of finite local complexity, unique ergodicity of $X_{D}$
	is equivalent to the existence of \emph{uniform pattern frequencies} in $D$,
	see e.g.\  \cite[Thm.\ 2.7]{LeMo02} or \cite[Prop.\ 2.32]{MR}
	for a recent generalisation.
\end{Remarks}

\subsection{Existence of the integrated density of states}
Due to the uniform boundedness assumption of the random variables $(\omega_{p})_{p\in D}$ in \textbf{(A0)}, we will choose the colour space as
\begin{equation} \label{A-w}
 	\AA = [-w,w] \subset\R
\end{equation}
right away.
In this way we obtain an operator-valued function $\hat X_{D} \ni P^{\omega} \mapsto H_{P^{\omega}}$ on the coloured translation orbit of $D$.
Assumptions \textbf{(A0)} ensure that, given any $P^{\omega} \in \hat X_{D}$ and any Borel measurable function $F: \R \rightarrow \C$ for which there exist constants $\gamma,\tau >0$ such that
\begin{equation}
 	\label{BLM-cond}
	|F(E)| \le \gamma\min\{1, \e^{-\tau E}\} \qquad \text{for every~} E\in\R,
\end{equation}
the operator $F(H_{P^{\omega}})\,\chi_{\L_{L}(y)}$ is trace class for every  $L>0$ and every $y\in\Rd$, where $\chi_{\L_{L}(y)}$ stands for the characteristic function of the cube $\L_{L}(y)$.
Moreover, it follows from \cite[Thm.\ 1.14(i)]{BLM} that the operator $F(H_{P^{\omega}})$ has a bounded continuous integral kernel $f(H_{P^{\omega}}) \in C(\Rd\times\Rd) \cap{\mathrm L}^{\infty}(\Rd\times\Rd)$. Thus, we infer the representation
\be\label{trace} \Tr \big( F(H_{P^{\omega}})\,\chi_{\L_{L}(y)} \big)
= \int_{\L_{L}(y)}f(H_{P^{\omega}})(x,x) \,\d x \ee
for its trace, which follows e.g.\ from \cite[Cor.\ 1.16 and 1.18]{BLM}).
The following ergodic theorem will be the main technical result in proving existence of the integrated density of states.

\begin{thm}
\label{vague-conv}
Let $D$ be a Delone set, $\mu$ an ergodic Borel probability measure on its hull $X_{D}$, $F: \R \rightarrow \C$ a Borel measurable function obeying \eqref{BLM-cond}, and assume \Az and \Ao.  Then,
\begin{nummer}
\item  \label{vague-conv-gen}
 there exists a measurable subset $Y \subseteq X_{D}$ (depending on $F$) of full probability, $\mu(Y) =1$, and for every ${P} \in Y$ there exists a measurable subset $\Xi_{P} \subseteq \Omega_{{P}}$  of full probability, $\P_{{P}}(\Xi_{{P}}) =1$, such that for every $\omega\in\Xi_{{P}}$ and every $y\in\Rd$ the limit
\begin{align}
	\lim_{L\rightarrow \infty}  \frac{1}{L^{d}} \; \Tr \left(F(H_{{P}^{\omega}})\chi_{\L_{L}(y)}\right) & =
	\int_{\hat X_{D}}  f(H_{\tilde{P}^{\tilde{\omega}}})(0,0) \, \d\hat\mu(\tilde{P}^{\tilde{\omega}}) \nonumber\\
	& = \int_{X_D} \left( \int_{\Omega_{\tilde P}} f(H_{\tilde{P}^{\tilde{\omega}}})(0,0) \, \d\P_{\tilde P}(\tilde\omega) \right) \d\mu(\tilde P) \label{intX_D}
\end{align}
exists and is independent of ${P} \in Y$, $\omega\in\Xi_{{P}}$ and $y\in\Rd$. Here,  $\hat\mu$ is given by Theorem~\ref{ergthm}.
\item \label{vague-conv-cont}
	if, in addition, $X_{D}$ is uniquely ergodic, $F\in C(\R)$ and \At is assumed, then \itemref{vague-conv-gen} holds with $Y =X_{D}$. In particular, \eqref{intX_D} holds with ${P}=D$.
\end{nummer}
\end{thm}

\begin{proof}
Let $\{U_a\}_{a\in\Rd}$ be the family of unitary operators on $\Lp{\Rd}$ associated to translations, that is, $U_{a}\psi := \psi(\cdot - a)$ for every $\psi\in L^{2}(\Rd)$ and $a\in\Rd$. (In the case $d=2$ and ${\textbf A} \neq 0$ we use magnetic translations.) Recalling the shifts \eqref{taudef} between the probability spaces, we find
\be
	\label{delone-trans}
	U_a \hdom U_a^{*} = H_{(a+D)^{\tau_{a}(\omega)}}= H_{a+ D^\omega},
\ee
and thus $F(H_{a+ \dom}) = U_{a} F(\hdom) U_{a}^{*}$. In turn, this implies
\begin{equation}
	\label{kernel-shift}
 	f(H_{a+ \dom}) =  	f({\hdom})(\cdot -a, \cdot -a)
\end{equation}
for the corresponding continuous integral kernels,
and in particular $f(H_{\dom})(x,x) = f(H_{-x +\dom})(0,0)$
for every $x\in\Rd$.
Now we define the map
\begin{equation} \label{Phidef}
 	\Phi :  \begin{array}{lcl} \hat X_{D} & \longrightarrow &\mathbb{C}, \\
 P^{\omega} & \longmapsto & \Phi(P^{\omega}):= f(H_{P^{\omega}})(0,0),
\end{array}
\end{equation}
and conclude from \eqref{trace} that
\begin{equation}
 	\Tr \big( F(\hdom)\,\chi_{\L_{L}(y)} \big) = \int_{\Lambda_{L}(-y)} \Phi(x+ D^{\omega}) \, \d x.
\end{equation}
By Lemma \ref{contker} below,  the map $\Phi$ is bounded and measurable, resp.\ continuous, under the hypotheses of part~(i), resp.\ part~(ii), of the theorem.
Thus, for fixed $y \in\Rd$, the claim follows from Theorem~\ref{ergthm}. Moreover, the limit does not depend on $y \in\Rd$ because the function $\Rd \ni x \mapsto \Phi(x+ D^{\omega})$ is bounded and because for every $y, y' \in\Rd$ the Lebesgue volume of the symmetric difference  $\Lambda_{L}(y) \sydi \Lambda_{L}(y')$ behaves like
$\mathcal{O}(L^{d-1})$ as $L\to\infty$.
\end{proof}

\noindent
The above proof rests upon

\begin{lem}\label{contker}
\begin{nummer}
\item Under the hypotheses of Theorem~\ref{vague-conv-gen}, the map $\Phi$ from \eqref{Phidef} is bounded and measurable.
\item Under the hypotheses of Theorem~\ref{vague-conv-cont}, it is even continuous.
\end{nummer}
\end{lem}

\begin{proof}
Let $P^{\omega} \in \hat X_{D}$. As $F$ satisfies the condition  \cite[Eq.\ (1.20)]{BLM}, the integral kernel has the  representation \cite[Thm.\ 1.14]{BLM}
\begin{align}\label{ktf}
	f(H_\pom)(0,0) &= \Big\langle k_t^{H_\pom}(\cdot,0), \e^{2tH_\pom}F(H_\pom) \,k_t^{H_\pom}(\cdot,0) \Big\rangle , \\[1ex]
	& = \sum_{n,m\in\N}  \big(\e^{-tH_\pom}\psi_{n}\big)(0) \; \langle\psi_{n}, \e^{2tH_\pom}F(H_\pom) \,\psi_{m}\rangle \notag\\[-1ex]
	& \hspace*{1.7cm} \times \overline{\big(\e^{-tH_\pom}\psi_{m}\big)(0)} \label{ktf-meas}
\end{align}
where $\{\psi_{n}\}_{n\in\N}$ is an orthonormal basis in $\mathrm{L}^{2}(\Rd)$, $t \in ]0, \tau/2[$ is arbitrary and
\be\label{kt}
	\Rd\times\Rd \ni (x,y) \mapsto k_t^{H_\pom}(x,y):=\frac{\e^{-\abs{x}^2/2t}}{(2\pi t)^{d/2}}\int
	\e^{-S_t(A,V_\pom;b)}\,\d\mu_{x,y}^{0,t}(b),
\ee
is the continuous heat kernel of $H_\pom$. Here, we have expressed the heat kernel $k_{t}^{H_{\pom}}$ in the Feynman--Kac representation, where $\mu_{x,y}^{0,t}$ is the standard Brownian-bridge probability measure on all continuous paths $b$, which start at $x$ at time zero and end at $y$ at time $t$, and
\begin{equation}
	\label{action}
 	S_{t}(A,V;b) := \i \int_{0}^{t} A(b(s)) \cdot \d b(s) +
  \frac{\i}{2}\,\int_{0}^{t} (\nabla\cdot A)(b(s)) \,\d s+
  \int_{0}^{t} V(b(s))\,\d s
\end{equation}
is the Euclidian action functional. The first integral on the right-hand side of
\eqref{action} is a stochastic line integral to be understood in the
sense of It\^{o}. The other two integrals are meant in
the sense of Lebesgue.

\emph{Part}~(i). \;
By \cite[Thm.~1.10]{BLM}, the image  $\e^{-tH_{\pom}}\psi$  of any $\psi\in\mathrm{L}^{2}(\Rd)$ under the semigroup has a continuous representative in $\mathrm{L}^{2}(\Rd)$. Thus, $(\e^{-tH_{\pom}}\psi) (0) =\lim_{\varepsilon\downarrow 0} \langle \eta_{\varepsilon},\e^{-tH_{\pom}}\psi\rangle$, where $(\eta_{\varepsilon})_{\varepsilon > 0} \subset C_{c}^{\infty}(\Rd)$ is a compactly supported, non-negative approximation of the Dirac delta function at $0\in\Rd$. Since $H_{\pom}$ is bounded from below uniformly in $\pom$, measurability of $\Phi$ follows \eqref{ktf-meas} and measurability of the map
\begin{equation}
	\label{G-meas}
 	\hat X_{D} \ni \pom \mapsto \langle \phi, G(H_{\pom}) \psi\rangle \qquad \text{for every~} \phi,\psi\in \mathrm{L}^{2}(\Rd),
\end{equation}
where  $G:\R\rightarrow\C$ is any bounded Borel measurable function. By the functional calculus, this holds if and only if the map \eqref{G-meas}
is measurable for indicator functions $G=\chi_{B}$, $B \subseteq \R$ any Borel set. In other words,
measurability of the unbounded self-adjoint operator-valued map $\hat X_{D} \ni \pom \mapsto H_{\pom}$ according to \cite[Def.\ V.1.3]{CaLa90} implies that $\Phi$ is measurable.
For $E>0$ consider the truncated kinetic-energy operator $H_{0}^{E} := H_{0} \,\chi_{[0,E]}(H_{0})$, which is bounded, and let $H^{E}_{\pom} := H_{0}^{E} + V_{\pom}$. We note the strong resolvent convergence of $H^{E}_{\pom}$ to $H_{\pom}$ as $E\to\infty$. Therefore \cite[Prop.\ V.1.4]{CaLa90} ensures that measurability of the map $\pom \mapsto H^{E}_{\pom}$ for every $E>0$ implies measurability of the map $\pom \mapsto H_{\pom}$. But since $ H^{E}_{\pom}$ is bounded, this means that it suffices to show measurability of the map
\begin{equation}
	\label{V-meas}
 	\hat X_{D} \ni \pom \mapsto \langle \phi, V_{\pom} \psi\rangle
\end{equation}
for every $\phi,\psi\in \mathrm{L}^{2}(\Rd)$. In fact, by the boundedness of $V_{\pom}$, it suffices to prove measurability of \eqref{V-meas} for every $\phi,\psi \in C_{c}^{\infty}(\Rd)$. To show this, let $\varepsilon>0$, consider the mollified single-site potential
$u^{\varepsilon} := u * \eta_{\varepsilon} \in C_{c}^{\infty}(\Rd)$ and define $V_{\pom}^{\varepsilon} := \sum_{p\in P} \omega_{p} u^{\varepsilon}(\,\pmb\cdot\, -p)$. By definition of the vague topology, the map $\hat X_{D} \ni \pom \mapsto V_{\pom}^{\varepsilon}(x)$ is continuous for every $x\in\Rd$. Now, let $\phi,\psi \in C_{c}^{\infty}(\Rd)$. Using a bound like \eqref{V-bound} for $V^{\varepsilon}_{\pom}$, dominated convergence yields continuity -- and therefore measurability -- of the map $\pom \mapsto \langle \phi, V_{\pom}^{\varepsilon} \psi\rangle$ for every $\varepsilon>0$. Finally, we conclude $\langle \phi, V_{\pom} \psi\rangle = \lim_{\varepsilon\downarrow 0} \langle \phi, V_{\pom}^{\varepsilon} \psi\rangle$ because of $\lim_{\varepsilon\downarrow 0} V_{\pom}^\varepsilon = V_{\pom}$ almost everywhere on $\Rd$ and another application of dominated convergence using again a bound like \eqref{V-bound}. Therefore the map \eqref{V-meas} is measurable for every $\phi,\psi \in C_{c}^{\infty}(\Rd)$ and
Part~(i) is proven.

\emph{Part}~(ii). \;
We use the representation \eqref{ktf} and suppose that the sequence $(P_n^{\omega_n})_{n\in\N} \subset \hat X_D$ converges to $Q^{\omega} \in \hat{X}_{D}$ in the vague topology. We use the abbreviations
\be
	H_n := H_0+V_{P_{n}^{\omega_{n}}}, \qquad  H:= H_0+ V_{Q^{\omega}}
\ee
and estimate with the triangle and the Cauchy-Schwarz inequality
\begin{align}
	\label{kercont}
 	\big|\Phi(P_{n}^{\omega_{n}})& - \Phi(Q^{\omega}) \big|  \nonumber\\
	&	= \big| f(H_n)(0,0)-f(H)(0,0) \big|  \nonumber\\
	& \le \big| \big\langle k_t^{H_n}(\cdot,0), \e^{2tH_n}F(H_n)
			\big( k_t^{H_n}(\cdot,0)- k_t^{H}(\cdot,0) \big) \big\rangle \big| \nonumber \\
	& \quad + \big| \big\langle k_t^{H_n}(\cdot,0), \big( \e^{2tH_n}F(H_n) - \e^{2tH}F(H) \big)
			\,k_t^{H}(\cdot,0) \big\rangle \big| \nonumber\\
	& \quad + \big| \big\langle \big( k_t^{H_n}(\cdot,0)- k_t^{H}(\cdot,0) \big),
	      \e^{2tH}F(H) \, k_t^{H}(\cdot,0)	 \big\rangle \big| \nonumber \\
	& \le \Big[ \big\| \e^{2tH_n}F(H_n)\big\|\, \big\|k_t^{H_n}(\cdot,0)\big\|  +
			\big\| \e^{2tH}F(H) \big\|\, \big\| k_t^{H}(\cdot,0) \big\| \Big]  \nonumber\\
	&	\qquad	\times \big\| k_t^{H_n}(\cdot,0)- k_t^{H}(\cdot,0) \big\|    \nonumber\\
	& \quad  + \big\| k_t^{H_n}(\cdot,0)\big\| \,
			\big\| \big( \e^{2tH_n}F(H_n) - \e^{2tH}F(H) \big) \, k_t^{H}(\cdot,0)\big\|.
\end{align}
From \eqref{kt}, \eqref{V-bound} and \eqref{A-w} we deduce the bound
\be\label{kb}
	\big| k_t^{H_n}(x,0) \big| \le \frac{\e^{-\abs{x}^2/2t}}{(2\pi t)^{d/2}} \; \e^{t wv_{0}}
\ee
for all $x\in\Rd$ and $t>0$, uniformly in $P_{n}^{\omega_{n}} \in \hat{X}_{D}$. Hence, we have $\sup_{n\in\N} \|  k_t^{H_n}(\cdot,0)\| < \infty$. Furthermore, \eqref{BLM-cond}, the fact that $t \in ]0, \tau/2[$ and $H_{\pom}$ is uniformly bounded below in $\pom \in \hat{X}_{D}$ implies
$\e^{2tH_n}F(H_n) = G(H_{n})$ for every $n\in\N$ and $\e^{2tH}F(H) = G(H)$, where $G\in C(\R)$ is some \emph{bounded} continuous function. In particular, $\sup_{n\in\N} \big\|\e^{2tH_n}F(H_n)\big\|
< \infty$ holds. Thus, we infer from \cite[Thm.\ VIII.20(b)]{RSI} that the following two conditions are sufficient for the vanishing of the
left-hand side of  \eqref{kercont} as $n\to\infty$: (a)~ convergence
$k_t^{H_n}(\cdot,0) \rightarrow k_t^{H}(\cdot,0)$ in $\Lp\Rd$
and ~(b)~ convergence $H_n \rightarrow H$ in strong-resolvent sense.

We will first verify condition (b).
By \cite[Thm.\ VIII.25(a)]{RSI} it is sufficient to prove
\be\label{srs}
	\lim_{n\rightarrow \infty}  \norm{ (H_n - H)\varphi} = 0
\ee
for all $\varphi\in C_{c}^{\infty}(\Rd)$.
So fix an arbitrary function $\varphi\in C_c^\infty(\Rd)$ and let
$0<\varepsilon <r/2$.
Then there is $K \subset\Rd$ compact (but large enough) such that
$\| \chi_{K^{c}} \chi_{\supp \varphi}\| < \varepsilon$.
The estimate
\begin{align}
	\label{suma}
	\|( H_{n} - H)\varphi\| & =
		\| ( V_{Q^{\omega}} - V_{P_{n}^{\omega_{n}}}) \varphi \|  \nonumber\\
	& \le  \| ( V_{Q^{\omega}} - V_{P_{n}^{\omega_{n}}})
		\chi_{K^c} \varphi \|  + \| ( V_{Q^{\omega}} - V_{P_{n}^{\omega_{n}}})
		\chi_{K} \varphi\| \nonumber\\
	& \le 2\varepsilon wv_{0} \, \|\varphi\| + \| ( V_{Q^{\omega}} - V_{P_{n}^{\omega_{n}}})
		\chi_{K} \|_\infty \; \|\varphi\|
\end{align}
holds for every $n\in\N$.
We define the thickened compact
\begin{equation}
	\label{thickening}
	K' := (K)_{\diam\supp u} := \bigcup_{x\in K} \L_{2\diam\supp u}(x)
\end{equation}
and its coloured version $\hat{K}' := K' \times \AA$.
Convergence of $P_{n}^{\omega_{n}}$ to $Q^{\omega}$ in the vague toplogy implies \cite[Lemma 2.8]{MR} that for every $\tilde{\varepsilon}>0$ there exists $n_{0}\in\N$ such that
for all $n \ge n_{0}$ we have
\begin{equation}
	\label{close}
 	P_{n}^{\omega_{n}} \cap \hat{K}' \subset (Q^{\omega})_{\tilde\varepsilon} \qquad\text{and} \qquad
	Q^{\omega} \cap \hat{K}' \subset (P_{n}^{\omega_{n}})_{\tilde\varepsilon}.
\end{equation}
Here, the thickening on the product space $\Rd\times\AA$ is defined in terms of cubes with respect to the maximum norm of the norms on $\Rd$ and $\AA$. As a consequence of \eqref{close} there is a one-to-one correspondence $p_{n,j} \leftrightarrow q_{j}$ between points of $P_{n}$ and $Q$ whenever one of the points lies in $K'$. Each of those pairs has the property that $|p_{n,j}-q_{j}| < \tilde\varepsilon$ and
$|\omega_{n}(p_{n,j}) - \omega(q_{j})| < \tilde\varepsilon$. We write $\mathcal{J}$ for the
finite index set labelling those corresponding points. Then, setting
$\tilde\varepsilon := \varepsilon/|\mathcal{J}|$, we can estimate
\begin{align}\label{vconv}
	\big\| ( V_{Q^{\omega}} - V_{P_{n}^{\omega_{n}}}) \chi_{K} \big\|_\infty    & \le
		\sum_{j\in\mathcal{J}} \big\| \omega(q_{j}) \, u( \boldsymbol\cdot - q_{j})
		- \omega_{n}(p_{n,j}) \, u( \boldsymbol\cdot - p_{n,j})\big\|_{\infty}  \nonumber \\
	& \le \sum_{j\in\mathcal{J}} \Big [ |\omega(q_{j}) - \omega_{n}(p_{n,j})| \,
			\| u( \boldsymbol\cdot - q_{j}) \|_{\infty} \nonumber\\
	& \qquad\quad + |\omega_{n}(p_{n,j})| \, \| u( \boldsymbol\cdot - q_{j})
			- u( \boldsymbol\cdot - p_{n,j}) \|_{\infty}\Big] \nonumber\\
	& \le \varepsilon  \big[ \| u\|_{\infty} + w \| |\nabla u|\|_{\infty} \big]
\end{align}
for every $n\ge n_{0}$.
This bound and \eqref{suma} complete the proof of condition (b).

In the rest of this proof we verify condition (a). We want to prove
that
\be
	\label{cond1-stat}
	\big\| k_t^{H_n}(\cdot,0)- k_t^{H}(\cdot,0)\big\|^{2} = \int_{\Rd}
	\big| k_t^{H_n}(x,0)- k_t^{H}(x,0) \big|^{2} \, \d x  \;\,\longrightarrow 0
\ee
as $n\to\infty$.
Using the representation \eqref{kt}, we infer
\begin{multline}
	\label{k1}
	\big| k_t^{H_n}(x,0) - k_t^{H}(x,0) \big|   \\
	\leq  \frac{\e^{-\abs{x}^2/(2t)}}{(2\pi t)^{d/2}}
			\int \Big| \e^{-S_t(0,V_{P_{n}^{\omega_{n}}};b)} -
			\e^{-S_t(0,V_{Q^{\omega}};b)} \Big| \,  \d\mu_{x,0}^{0,t}(b)
\end{multline}
for all $n\in\N$.
The elementary inequality $|\e^\xi - \e^{\xi'}|\leq \abs{\xi-\xi'}\e^{\max\{\xi,\xi'\}}$ for all $\xi,\xi'\in\R$ then allows to estimate the integral in \eqref{k1} from above by
\begin{multline}
	\label{decoup}
 	\e^{twv_{0}}   \int \big| S_t(0,V_{P_{n}^{\omega_{n}}};b) -
		S_t(0,V_{Q^{\omega}};b) \big| \,\d\mu_{x,0}^{0,t}(b) \\
	 \leq  \e^{twv_{0}} \int S_t(0,|V_{P_{n}^{\omega_{n}}} - V_{Q^{\omega}}|;b) \, \d\mu_{x,0}^{0,t}(b).
\end{multline}
Now, for every given $0 < \varepsilon < r/2$ there exists a length $\ell >0$ (depending on $t$ but not on $x$) such that with $K:= \overline{\Lambda_{\ell}(0)}$ we have
\begin{equation}
	\label{compl-est}
 	\int S_t(0,\chi_{K^{c}};b) \, \d\mu_{x,0}^{0,t}(b) < \varepsilon (1 + |x|^{4}).
\end{equation}
This estimate can be derived from an explicit calculation of the Brownian-bridge expectation after applying the Chebyshev-Markov inequality, see e.g.\ \cite[Eq.\ (2.15)]{BLM}. We define the corresponding thickened set $K'$ as in \eqref{thickening} and $\hat{K}':= K' \times\AA$. Convergence of $P_{n}^{\omega_{n}}$ to $Q^{\omega}$ then implies the existence of $n_{0}\in\N$ such that for all
$n \ge n_{0}$ the estimate \eqref{vconv} holds. This and \eqref{compl-est} imply
\begin{equation}
	\label{cond1-final}
 	\int\! S_t(0,|V_{P_{n}^{\omega_{n}}} - V_{Q^{\omega}}|;b)  \, \d\mu_{x,0}^{0,t}(b)
	\le \varepsilon \big[ t ( \| u\|_{\infty} + w \| |\nabla u|\|_{\infty} )
	+ wv_{0}(1 + |x|^{4}) \big]
\end{equation}
for all $n \ge n_{0}$. Therefore \eqref{cond1-stat} follows from \eqref{k1} -- \eqref{cond1-final}.
\end{proof}


Since $H_{P^{\omega}}$ is uniformly lower semi-bounded for $P^{\omega} \in \hat{X}_{D}$, we have
\begin{equation}
	\label{tight}
 	\chi_{]-\infty,E]}(H_{P^{\omega}})=\chi_{[E_0,E]}(H_{P^{\omega}})
\end{equation}
for some $E_0 \in\R$, which depends on $r$ but not on the point set $P$ or its random colouring $\omega$. Furthermore, the spectral projection
$\chi_{]-\infty,E]}(H_{P^{\omega}})$ has an integral kernel	$p_E(H_{P^{\omega}}) \in
C(\Rd\times\Rd) \cap \mathrm{L}^{\infty}(\Rd\times\Rd)$, see e.g.\ \cite[Thm.\ 1.14(i)]{BLM}.

\begin{defn}\label{defIDS}
	Let $D$ be a Delone set and $\mu$ an ergodic Borel probability measure on its hull $X_{D}$. Let $\hat\mu$ be given by Theorem~\ref{ergthm}.
	The \emph{integrated density of states} w.r.t.\ $\mu$ of the family of operators
	$\hat{X}_{D} \ni P^{\omega} \mapsto H_{P^{\omega}}$
	is the right-continuous non-decreasing function
	\be
		\label{intIDS}
		\R \ni E \mapsto \nu_{D}(E) :=\int_{\hat X_D} \,p_E(H_{P^{\omega}})(0,0)\, \d\hat\mu(P^{\omega})
	\ee
	with values in $[0,\infty[$.
\end{defn}

\begin{rem}
If $X_{D}$ is uniquely ergodic, then the integrated density of states is unique.
\end{rem}

Given $P^{\omega} \in \hat{X}_{D}$, the mappping
\be
	C_{c}(\Rd) \ni F \mapsto \frac{1}{L^d} \,\tr \Big( F(H_{P^{\omega}})\chi_{\L_{L}(y)} \Big)
\ee
is a well defined continuous positive linear functional on $C_c(\R)$ (equipped with the inductive limit topology), see e.g.\ \cite[Section 2.4.2]{H}. By the Riesz-Markov representation theorem, it defines a unique right-continuous non-decreasing function
\be
	\label{fvids}
	\R \ni E \mapsto \nu_{\pom\!,L,y}(E) = \frac{1}{L^{d}} \,
	\tr \left( \chi_{]-\infty,E]}(H_{P^{\omega}})\chi_{\L_{L}(y)}\right),
\ee
the \emph{finite-volume integrated density of states}, such that
\be
	\label{tr}
	\frac{1}{L^d} \, \tr \Big( F(H_{P^{\omega}})\chi_{\L_{L}(y)} \Big) = \int_\R F(E)\, \d\nu_{\pom\!,L,y}(E) .
\ee
Now, we apply Theorem~\ref{vague-conv} to justify the terminology integrated density of states for $\nu_{D}$.

\begin{cor}\label{eids}
Let $D$ be a Delone set and assume \Az and \Ao. Let $\mu$ be an ergodic Borel probability measure on the hull $X_{D}$ and let $\hat\mu$ be given by Theorem~\ref{ergthm}. Then,
\begin{nummer}
\item  \label{ids-conv-gen}
 there exists a measurable subset $Y \subseteq X_{D}$ of full probability, $\mu(Y) =1$, and for every ${P} \in Y$ there exists a measurable subset $\Xi_{{P}} \subseteq \Omega_{{P}}$  of full probability, $\P_{{P}}(\Xi_{{P}}) =1$, such that for every $\omega\in\Xi_{{P}}$ and every $y\in\Rd$ we have
	\begin{equation}
		\label{eids-eq}
 		\lim_{L\to\infty} \nu_{\pom\!,L,y}(E) = \nu_{D}(E)
	\end{equation}
	for every $E\in\R$, independently of ${P} \in Y$, $\omega\in\Xi_{{P}}$ and $y\in\Rd$.
\item \label{ids-conv-cont}
	if, in addition, $X_{D}$ is uniquely ergodic and \At holds, then one may choose $Y =X_{D}$ in
	\itemref{ids-conv-gen}, provided $E$ is a point of continuity of $\nu_{D}$. In particular, \eqref{eids-eq} holds for ${P}=D$ at continuity points of $\nu_{D}$.
\end{nummer}
\end{cor}

\begin{Remarks}
\item We recall Remark~\ref{on-unique} for conditions on $D$ ensuring unique ergodicity of $X_{D}$.
\item Without unique ergodicity of $X_{D}$ one cannot expect the corollary to hold  with $Y=X_{D}$.
 	We refer to Appendix~\ref{sec:example} for an example.
\item It is not clear whether Part~\itemref{ids-conv-cont} of the corollary can be extended to hold for all energies as Part~\itemref{ids-conv-gen}.
    \item Similar results have been obtained for discrete aperiodic structures, see \cite{LS03,LS06,LPV07,LMV08}.
\end{Remarks}

\begin{proof}[Proof of Corollary~\ref{eids}]
Part (i). \quad The proof is the same as that of \cite[Cor.\ 5.8]{K} with Theorem~\ref{vague-conv-gen} playing the role of \cite[Prop.~5.2]{K}. For this argument to work it is crucial that one can choose $F=\chi_{]-\infty,E]}$ in Theorem~\ref{vague-conv-gen}.

Part (ii). \quad
Given $F\in C_c(\R)$ and any ${P} \in X_{D}$, Theorem~\ref{vague-conv-cont} ensures the existence of a measurable set $\Xi^{F}_{P} \subset \Omega_{D}$ with $\P_{P}(\Xi^{F}_{P})=1$ such that
\be
	\label{limF}
	\lim_{L\rightarrow\infty} \int_\R F(E) \, \d\nu_{\pom\!,L,y}(E)
	= \int_{\hat{X}_{D}} f(H_{\tilde{P}^{\tilde\omega}})(0,0)\, \d\hat\mu(\tilde{P}^{\tilde\omega})
\ee
for every $y\in\Rd$ and every $\omega\in\Xi_{P}^F$.
The functional calculus also holds for integral kernels
\begin{equation}
	\label{stieltjes}
	f(H_{\tilde{P}^{\tilde\omega}})(0,0) = \int_{\R} F(E)\, \d p_{E}(H_{\tilde{P}^{\tilde\omega}})(0,0),
\end{equation}
where the right-hand side is to be understood as a Lebesgue-Stieljes integral with respect to the
non-decreasing function $E \mapsto p_{E}(H_{\tilde{P}^{\tilde\omega}})(0,0)$,
see e.g.\ \cite[Cor.\ 1.18]{BLM}.
Since constant functions over compact subsets of $\R$ are integrable w.r.t.\ $\d p_{E}(H_{\tilde{P}^{\tilde\omega}})(0,0)$, Fubini's theorem gives
\be
	\label{limF2}
	\lim_{L\rightarrow\infty} \int_\R F(E) \, \d\nu_{\pom\!,L,y}(E)
	= \int_{\R} F(E)\, \d\nu_{D}(E)
\ee
for every $y\in\Rd$ and every $\omega\in\Xi_{P}^F$.

Next, we show that \eqref{limF2} holds for a set of full
$\P_{P}$-probability \emph{independently} of $F\in C_c(\R)$.
To this end, we need a particular countable dense (w.r.t.\ $\|\pmb\cdot\|_{\infty}$) subset of
$C_c(\R)$. For $K\in\N$ let $\tilde{\mathcal{D}}_{K}\subset C([-K,K])$ be countable and dense.
Given any $\tilde{f} \in \tilde{\mathcal{D}}_{K}$, let $f_{\tilde{f}}\in C_{c}(\R)$ be an extension of $\tilde{f}$ with $\supp f_{\tilde{f}} \subseteq [-K-1, K+1]$ and define
$\mathcal{D}_{K} := \{ f_{\tilde{f}} \in C_{c}(\R) : \tilde{f} \in \tilde{\mathcal{D}}_{K}\}$. Then
$\mathcal{D} := \bigcup_{K\in\N} \mathcal{D}_{K}$ is countable and dense in $C_{c}(\R)$. For every $K\in\N$ let
$ 0 \le \psi_{K} \in C_{c}(\R)$ with $\psi_{K}|_{[-K,K]} =1$ and define
\begin{equation}
	\Xi_{P} := \Big(\bigcap_{F\in \mathcal{D}}\Xi_{P}^F \Big) \cap  \Big(\bigcap_{K\in\N}\Xi_{P}^{\psi_{K}} \Big)
\end{equation}
so that $\P_{P}(\Xi_{P}) =1$ and
\eqref{limF2} holds simultaneously for all $F\in\mathcal{D}$, all $\psi_{K}$, $K\in\N$, and all $\omega\in\Xi_{P}$.
The following approximation argument extends the validity of \eqref{limF2} to all $F\in C_{c}(\R)$ and
all $\omega\in\Xi_{P}$. Given $F\in C_{c}(\R)$ there exists a sequence $(F_{n})_{n\in\N} \subset \mathcal{D}$ and $K\in \N$ such that $\|F - F_{n}\|_{\infty} \rightarrow 0$ as $n\to\infty$ and $\psi_{K} F=F$, $\psi_{K} F_{n} = F_{n}$ for all $n\in\N$. Therefore we get
\begin{align}
 	\bigg| \int_\R & F(E) \, \d\nu_{\pom\!,L,y}(E)
	- \int_{\R} F(E)\, \d\nu_{D}(E) \bigg|  \nonumber\\
	&\le \|F-F_{n}\|_{\infty} \bigg( \int_\R \psi_{K}(E) \, \d\nu_{\pom\!,L,y}(E)
	 + \int_{\R} \psi_{K}(E)\, \d\nu_{D}(E)\bigg)	 \nonumber\\
	& \quad +  \bigg| \int_\R  F_{n}(E) \, \d\nu_{\pom\!,L,y}(E)
	- \int_{\R} F_{n}(E)\, \d\nu_{D}(E) \bigg|
\end{align}
for every $n\in\N$, $L>0$, $y\in\Rd$ and every $\omega\in\Xi_{P}$. Using \eqref{limF2}, we conclude
\begin{multline}
 	\limsup_{L\to\infty} \bigg| \int_\R F(E) \, \d\nu_{\pom\!,L,y}(E)
	- \int_{\R} F(E)\, \d\nu_{D}(E) \bigg|   \\
	\le 2 \|F-F_{n}\|_{\infty} \int_{\R} \psi_{K}(E)\, \d\nu_{D}(E).
\end{multline}
The subsequent limit $n\to\infty$ shows that \eqref{limF2} holds for every $F\in C_{c}(\R)$, $y\in\Rd$ and
every $\omega\in\Xi_{P}$.
In other words, we have $\P_{P}$-a.s.\ vague convergence of the measures associated with
the non-decreasing functions $\nu_{\pom\!,L,y}$ to the measure associated with
$\nu_{D}$.
Due to the uniform lower boundedness \eqref{tight} of $H_{P^{\omega}}$, no mass can get lost towards $-\infty$ in this vague limit, and the claim follows.
\end{proof}


As in the usual alloy-type Anderson model with impurities situated on $\Z$, Theorem \ref{ergthm} allows to relate the growth points of $\nu_{D}$ to the spectrum of $H_{\dom}$. Our approach requires a condition on the geometric complexity of the Delone set $D$, see also \cite[Section 2.3]{MR}. The geometry of the Delone set is expressed in terms of patterns: we say that $Q\subset D$ is a \emph{pattern} of the Delone set $D$, if $Q=D \cap K$ for some compact subset $K\subset \Rd$. We will need the following definition concerning the appearance of patterns in the Delone set $D$.

\begin{defn}\label{def-upf}
The Delone set $D$ has \emph{uniform pattern frequency} if for any pattern $Q\subset D$ the quotient
\be
	\label{upf}
	\frac{\tilde\eta_{x,L}(Q)}{L^{d}}:= \frac{1}{L^d}\big| \left\{ \tilde Q\subset D\,:\, \exists y\in (x+\L_L) \,\,
	\mbox{ such that  }\,y+\tilde Q=Q \right\}\big|
\ee
converges uniformly in $x\in \R^{d}$ as $L \to\infty$.  Moreover, we say that $D$ has  \emph{strictly positive uniform pattern frequencies} if this limit is strictly positive.
\end{defn}

\begin{thm}\label{t:supp-ids}
	Let $D$ be a Delone set and assume \Az and \Ao. Let $\mu$ be an ergodic Borel probability measure on the hull $X_{D}$, $\hat\mu$ be given by Theorem~\ref{ergthm}.	Then,
	\begin{nummer}
	\item \label{growth-spec-as}
 there exists a measurable subset $Y \subseteq X_{D}$ of full probability, $\mu(Y) =1$, and for every ${P} \in Y$ there exists a measurable subset $\Xi_{{P}} \subseteq \Omega_{{P}}$  of full probability, $\P_{{P}}(\Xi_{{P}}) =1$, such that for every $\omega\in\Xi_{{P}}$ we have
	\be
		\label{growth-spec}
		\overline{\{E \in\R : E \text{~is a growth point of~} \nu_{D}\}}
		= \spec(H_{\pom}),
	\ee
	where the overbar denotes the closure in $\R$.
	\item if, in addition, the Delone set $D$ has strictly positive uniform pattern frequencies, $X_{D}$ is uniquely ergodic and \At holds, then one may choose $Y =X_{D}$ in \itemref{growth-spec-as}. In particular, \eqref{growth-spec} holds for ${P}=D$.
\end{nummer}
\end{thm}

\begin{proof}
Part (i). \quad The proof follows as in \cite[Theorem 3.1]{PF92}, taking into account \eqref{delone-trans}, Corollary \ref{eids}(i) and \eqref{intIDS}.\\
Part (ii). \quad
This result builds upon the discrete case in \cite[Lemma 6.3]{MR07}.  The continuous case requires additional technical arguments such as taking care of approximating $\chi_I(H_\pom)$ by a smooth function and constructing an orthogonal sequence of compactly supported functions to expand its trace.
To illustrate where the geometric assumptions on $D$ are needed, and for the reader's convenience, we give the proof.

We first show the inclusion $\subseteq$. Let $E$ be an element of the set on the l.h.s. of \eqref{growth-spec}. Then there exist continuity points  $a<E<b$ of $\nu_D$ such that Corollary \ref{eids}(ii) implies  for $I:=(a,b)$ the  limit representation
\be 0<\nu(I):=\int_I\d\nu_D= \lim_{L\rightarrow \infty} \frac{1}{L^d} \tr\left( \chi_{I}\left(H_\pom\right)\chi_{\L_L(y)}\right)
\ee
for every $y\in\Rd$, every $P\in X_D$ and $\P_P$-a.e. $\omega\in\Omega_P$. This shows $E\in\sigma(H_\pom)$.

Next, we show the inclusion $\supseteq$. Let $\epsilon>0$, let $P\in X_D$ and $\omega\in \supp \P_P$ such that $E\in I$ is in the spectrum of $H_\pom$. Then, by \cite[Theorem 7.22]{Weid}, there exists a sequence $(\varphi_n)_n\in C_c^\infty(\Rd)$ such that $C:=\displaystyle\liminf_{n\rightarrow\infty} \norm{\varphi_n}>0$ and $\norm{\left(H_\pom-E\right)\varphi_n}\rightarrow 0$.
Let $\delta>0$ to be specified later. There exists $n_0\in\N$ such that $\norm{\varphi_{n_0}}>C$ and
\be \norm{\left(H_\pom-E\right)\varphi_{n_0}}<\delta.\ee
Define $\tilde \varphi:=\varphi_{n_0}/\norm{\varphi_{n_0}}$, then
\be \norm{\left(H_\pom-E\right)\tilde\varphi}<\frac{\delta}{\norm{\varphi_{n_0}}}<\frac{\delta}{C}. \ee
Let $K_0\subset \Rd$ be the smallest compact set such that $\dist(\Rd\setminus K_0,\supp \tilde\varphi)>\delta_u$, where $\delta_u$ is the radius of a ball containing the support of the single-site potential $u$ in $V_\pom$. Denote by $G_0:= P\cap K_0$ the pattern of $P$ associated with $K_0$.  Because of the strict uniform pattern frequency property, we know that $G_0$ is repeated in $P$ infinitely many times, in fact, macroscopically often.
 Let $G_j=P\cap K_j$, $j\in\N_0$, be an enumeration of them obeying $G_j=v_j+G_0$, $K_j=v_j+K_0$ for translation vectors $v_j\in\Rd$,
$j\in \N_0$. For each $K_j$ there are at most $S\in\N$ supports $K_{j'}$, $S$ independent of $j$, that have non empty intersection with $K_j$. By dropping the overlapping ones, we extract a subsequence $\left(G_{j_l}\right)_l$ of patterns with pairwise disjoint supports. From now on we denote the subsequence $j_l$ again by $j$.  Denote by $\tilde\eta_L(G_0)$ the number of translated copies of $G_0$
in $\L_L$ that do not overlap, that is, $\tilde\eta_L(G_0)=\sharp\{ G_{j}\subset \L_L\}$.  We still have
\be\label{supf} \liminf_{L\rightarrow \infty} \frac{\tilde\eta_L(G_0)}{L^d}=:\eta>0, \ee

Consider the translated functions $\tilde\varphi_j:=U_{v_j}\tilde\varphi$, for $j\in\N_0$, where $U_a$ was defined above \eqref{delone-trans}.  We have that $\supp \tilde \varphi_j\subset K_j$ are pairwise disjoint, so $(\tilde\varphi_j)_{j\in\N_0}$ is an orthonormal sequence.

Then, the events
\be A_j:=\{ \tilde\omega\in\Omega_P\,:\, \abs{\tilde\omega_{p+v_j}-\omega_p}<\delta\,\, \text{ for all } p\in G_0\},\,\,j\in\N_0,\ee
of having a colouring in $G_j$ that is close to the one in $G_0$ given by $\omega$ are all independent by \Ao. Because of the $\Rd$-covariance of the probability measure, stated in \cite[Lemma 3.9(i)]{MR}, the fact that $\sharp G_j=\sharp G_0<\infty$ and that the $\omega_p$ are i.i.d., we have $\P_P(A_j)=\P_P(A_0)>0$ for all $j\in\N_0$, because $\omega\in\supp\P_P$.  The strong law of large numbers implies that there exists a set $\Xi_0\subset\Omega_P$ of full measure, $\P_P(\Xi_0)=1$, such that

\be\label{sln} \lim_{L\rightarrow \infty} \frac{1}{\tilde\eta_L(G_0)}\sum_{\substack{j\in\N_0: \\ K_j\subset\L_L }}\chi_{A_j}(\tilde\omega):=\P_P(A_0)>0 \ee
for all $\tilde\omega\in\Xi_0$.

Let $I_\epsilon:=(E-\epsilon,E+\epsilon)$ and take a function $F\in C_c(\R)$ such that $\supp F\subset I_\epsilon$, $ F\leq \chi_{I_\epsilon}$ and
\be F|_{I_{\epsilon/2}}=1. \ee

By Theorem \ref{vague-conv}(ii), there exists a set $\Xi_1\subset\Omega_P$ of full measure, $\P_P(\Xi_1)=1$, such that we have the following lower bound on $\nu_D(I_\epsilon)$,
\be\label{lowbn} \nu(I_\epsilon)=\int_{I_\epsilon}  \d\nu_D\geq \int_\R F\, \d\nu_D= \lim_{L\rightarrow \infty} \frac{1}{L^d}\tr \left( F(H_{P^{\tilde\omega}})\chi_{\L_L}\right), \ee
for all $\tilde\omega\in\Xi_1$.
Now, take $\tilde\omega\in \Xi_1\cap\Xi_2$. Then we obtain a lower bound for the trace by expanding it in the orthonormal sequence $(\tilde\varphi_j)_{j\in\mathbb N_0}$.  We have
\begin{align}  \tr \,\left( F(H_\tpom)\chi_{\L_L}\right) &  \geq \sum_{\substack{j\in\N_0: \\ K_j\subset\L_L }}  \angles{\tilde\varphi_j,  F(H_\tpom) \tilde\varphi_j} \nonumber \\
& \geq  \sum_{\substack{j\in\N_0: \\ K_j\subset\L_L }} \chi_{A_j}(\tilde\omega) \angles{\tilde\varphi_j,  F(H_\tpom) \tilde\varphi_j}\nonumber\\
& \geq  \sum_{\substack{j\in\N_0: \\ K_j\subset\L_L }} \chi_{A_j}(\tilde\omega) \angles{\tilde\varphi_j,  \chi_{I_{\epsilon/2}}(H_\tpom) \tilde\varphi_j}, \label{lowbf2}
\end{align}
where in the last inequality, we used that $\chi_{I_{\epsilon/2}}\leq F$.  The inner product on the r.h.s. of \eqref{lowbf2} can be bounded below according to
\begin{align}
 \angles{\tilde\varphi_j,  \chi_{I_{\epsilon/2}}(H_\tpom) \tilde\varphi_j} & =\norm{\tilde\varphi_j}^2-\norm{\chi_{\R\setminus I_{\epsilon/2}}(H_\tpom) \tilde\varphi_j}^2 \nonumber\\
& \geq 1-\left(\frac{\epsilon}{2}\right)^{-2}\norm{(H_\tpom - E) \tilde\varphi_j}^2.
\end{align}
For $\tilde\omega\in A_j$, we have
\begin{align} \norm{(H_\tpom - E) \tilde\varphi_j} & < \norm{\left( V_{v_j+\pom}-V_\tpom\right)\tilde\varphi_j} +\norm{(H_{v_j+\pom} - E) \tilde\varphi_j}   \nonumber\\
& < \delta \norm{\tilde \varphi} +\frac{\delta}{C}\norm{\tilde\varphi} \nonumber\\
& = \delta\left(1+\frac{1}{C}\right).
\end{align}
Combining this, \eqref{sln} and \eqref{supf}, we can separate the r.h.s. of \eqref{lowbf2} into a \emph{pattern frequency part} times an \emph{averaged randomness part}
\begin{align} \nu(I_{\epsilon}) & \geq \left( 1-4\left(1+\frac{1}{C}\right)^2 \frac{\delta^2}{\epsilon^2}\right) \, \liminf_{L\rightarrow \infty} \left( \frac{\tilde\eta_L(G_0)}{L^d} \cdot \frac{1}{\tilde\eta_L(G_0)} \sum_{\substack{j\in\N_0: \\ K_j\subset\L_L }}\chi_{A_j}(\tilde\omega)   \right)\\ \nonumber
& \geq \left( 1-4\left(1+\frac{1}{C}\right)^2  \frac{\delta^2}{\epsilon^2}\right) \, \eta \, \P_P(A_0).\end{align}
By taking $\delta>0$ small enough, we infer that $\nu(I_{\epsilon})>0$.
\end{proof}

\begin{rem}\label{rem:sptype} If the Delone set $D$ is periodic, any $P\in X_D$ can be written as $P=D+v$ where $v\in\Rd$ belongs to the fundamental cell of the lattice.  In this case one has the following standard result on the spectral type of the spectrum of $H_\dom$ \cite[Theorem 1]{KMa}: there exist sets $\Sigma_*\subset \R$, with $*=$ pp, ac, sc, such that
\be\label{sptype} \sigma_*(H_\pom)=\Sigma_* \quad\mbox{for all }P\in X_D \mbox{ and $\P_P$-a.e.}\,\, \omega\in\Omega_P,
\ee
where pp, ac, sc stands for pure point, absolutely continuous and singular continuous spectrum, respectively.

One way to prove \eqref{sptype} in some energy interval $I \subset \mathbb R$ for an arbitrary Delone set $D$, is to establish continuity of the function $X_D\ni P\mapsto f_*(P)=\mathbb E_P\left(\tr\, \left( \chi_I(H_\pom)\Pi_*\right)\right)$, where $\Pi_*$ is the orthogonal projection onto the spectral subspace corresponding to $*\in\{\text{pp},\text{ac},\text{sc}\}$. Note that if $D$ is periodic, $f_*$ is constant.
  Without knowing the continuity of $f_*$, one can only say that \eqref{sptype} holds for $\mu$-a.e. $P\in X_D$ and $\P_P$-a.e. $\omega\in \Omega_P$, as follows from the proof of \cite[Theorem 1]{KMa}.

We can prove, however, a weaker version of $\eqref{sptype}$ for the pure point spectrum at low energies in the case $H_0=-\Delta$, see Theorem \ref{t:dynlocas}.
\end{rem}


\section{Lifshitz tails and dynamical localization for Delone-Anderson operators}\label{s:LTdynloc}

In the first part of this section we prove that the integrated density of states of the Delone-Anderson operator with zero
magnetic field given by Eqs.\eqref{hdom} -- \eqref{ranpot1}, namely
\[ H_{D^{\omega}} = -\Delta+ V_{D^{\omega}},  \]
exhibits a Lifshitz-tail behaviour at the bottom of its spectrum. As an application, in the second part of this section, we will show dynamical localization for low energies in the spectrum of $H_{D^{\omega}}$ and investigate the size of the region of dynamical localization.
To prove dynamical localization we use the bootstrap multi-scale analysis for non-homogeneous systems \cite{GK1,RM}. Monotonicity of the Delone-Anderson operator in the random coupling constants is important in the proofs of both Lifshitz tails and dynamical localization. Moreover, we want to concentrate on Lifshitz tails of a purely quantum-mechanical character that are not affected by the details of the single-site distribution. Therefore we define the following assumptions:
\begin{itemize}
\item[\Aloc]
The single-site potential in \Az is non-negative and obeys
\be\label{u} u^-\chi_{\Lambda_{\epsilon_u}}\leq u \leq u^+\chi_{\Lambda_{\delta_u}}, \ee
for some constants $ 0<\epsilon_u\leq \delta_u <\infty $ and $0< u^-
\leq u^+ < \infty $.

The single-site distribution $\P^{(0)}$ in \Ao is absolutely continuous
with a bounded and continuous Lebesgue density $\rho$ with compact support 
\be\label{rhoplus0}
 0\in\mbox{supp }\rho \subseteq [0, w[.
\ee
\end{itemize}
As a consequence of \eqref{rhoplus0} we can determine the non-random spectrum $\Sigma$ of the family $(H_{\pom})_{P\in X_{D}, \omega\in\Omega_{P}}$ by a Borel--Cantelli type argument. This yields
\[ \sigma(H_\pom)=[0,\infty) \quad \mbox{for all } P\in X_D \mbox{ and~ $\mathbb{P}_{P}$-a.e. }\omega\in\Omega_P.\]

\begin{itemize}
\item[\Alt]
There exist constants $C_\rho,\beta>0$ such that the single-site probability density $\rho$ in \Aloc obeys
\be\label{rhobottom} \int_{0}^{\epsilon} \rho(v) \, \d v\geq C_\rho \epsilon^\beta \qquad   \forall \epsilon >0 \text{~small enough}. \ee
\end{itemize}


In what follows we denote by $C_{a,b,c,...}$ a positive constant depending on the parameters $a,b,c,...$.



\subsection{Lifshitz tails}

The integrated density of states $\nu_{D}$ of the family of operators $\hat{X}_{D} \ni P^{\omega} \mapsto H_\dom$ is exponentially suppressed near the bottom of the spectrum, because the occurrence of small eigenvalues requires a large-deviation event.

\begin{thm}\label{lt}
Let $D$ be a Delone set and let $\mu$ be an ergodic Borel probability measure on the hull $X_D$. 
Assume \Az with $B=0$, \Ao, \Aloc and \Alt. 
Then, the integrated density of states $\nu_D$ from \eqref{intIDS} exhibits a Lifshitz tail at the bottom of the spectrum, i.e.,
\begin{equation}
 	\lim_{E \downarrow 0} \; \frac{\ln |\ln \nu_D(E)|}{\ln E} = - \frac{d}{2}.
\end{equation}
\end{thm}

The proof relies on an averaged form of Dirichlet--Neumann bracketing over the hull in \eqref{ndbrack} below. For this purpose, we introduce an alternative version of the finite-volume integrated density of states \eqref{fvids},

\begin{defn} Given $P\in X_D$, $\omega\in\Omega_P$, $y\in\Rd$ and $L>0$, we consider the non-decreasing function
\be\label{fvn} \R \ni E \mapsto\tilde\nu_{\pom\!,L,y}^{\sharp}(E) :=\frac{1}{L^d} \;\tr \chi_{]-\infty,E]}(H_{\pom\!,L,y}^{\sharp}), \ee
where $H_{\pom\!,L,y}^\sharp$, $\sharp\in\{\rm{D,N}\}$, is the restriction of $H_\pom$ to the cube $\Lambda_L(y)$, with Dirichlet or Neumann boundary conditions. If $y=0$, we drop it from the notation and write simply $H_{\pom\!,L}^\sharp$ and $\tilde\nu_{\pom\!,L}^{\sharp}$.
\end{defn}

For any $F\in C_c(\Rd)$ we have
\be \int_\R F(E)\, \d\tilde \nu_{\pom\!,L,y}^{\sharp}(E)= \frac{1}{L^d} \;\tr F(H_{\pom\!,L,y}^\sharp).  \ee
The limit of the function $ \tilde \nu_{\pom\!,L,y}^{\sharp}$ as $L \to\infty$ is again $\nu_D$, as we recall in

\begin{prop}\label{approxnu}
Let $D$ be a Delone set and assume \Az and \Ao. Let $\mu$ be an ergodic Borel probability measure on the hull $X_{D}$. 
Then, there exists a measurable subset $Y \subseteq X_{D}$ of full probability, $\mu(Y) =1$, and for every ${P} \in Y$ there exists a measurable subset $\Xi_{P} \subseteq \Omega_{{P}}$  of full probability, $\P_{{P}}(\Xi_{P}) =1$, such that for every $\omega\in\Xi_{{P}}$ and every $y\in\Rd$ we have
\be\label{nu-tilde} \lim_{L\rightarrow \infty} \tilde \nu_{\pom\!,L,y}^{\sharp}(E)=\nu_D(E) \ee
for all continuity points $E$ of $\nu_D$, where $\nu_D$ is the integrated density of states from Corollary~\ref{eids}.
\end{prop}

\begin{rem}\label{nu_cont}
Under hypotheses \Az, \Ao and \Aloc, uniform Wegner estimates \cite{RMV,Kl} imply continuity of $\nu_D$. In this case, the convergence in \eqref{nu-tilde} holds for every $E\in\mathbb R$.
\end{rem}

\begin{proof}[Proof of Proposition~\ref{approxnu}]
Since the proof applies to both $\sharp=\rm D,N$, we omit the superscript from the notation here.
We infer from Corollary~\ref{eids}(i)  that $\nu_{\pom\!,L,y}$ converges vaguely to $\nu_D$, for $\P_P$-a.e. $\omega\in\Omega_P$.  On the other hand, \cite[Lemma 2.15]{H} gives
\be\label{appn} \lim_{L\rightarrow \infty} \abs{ \int_\R F(E) \,\d\tilde \nu_{\pom\!,L,y}(E)-\int_\R F(E) \,\d \nu_{\pom\!,L,y}(E)  }=0 \ee
for every $F\in C_c^\infty(\R)$.  Therefore, we have
 \be \lim_{L\rightarrow \infty}\int_\R F(E) \,\d\tilde \nu_{\pom\!,L,y}(E)=\int_\R F(E) \,\d\nu_D(E).
 \ee
 This, together with \cite[\S30, Exercise 3]{Ba}, yields vague convergence of the measure $\d\tilde \nu_{\pom\!,L,y}$ to $\d\nu_D$. Since all the measures are supported in $[0,\infty)$, no mass can get lost towards $-\infty$, and we obtain pointwise convergence of the distribution function $\tilde \nu_{\pom\!,L,y}$ to $\nu_D$ at continuity points of the latter.
\end{proof}
%
To establish Dirichlet--Neumann bracketing for $\nu_D$ we use the more convenient definition \eqref{fvn} for the finite-volume integrated density of states. 

\begin{lem}
	\label{ndbracklem} 
	Let $D$ be a Delone set and let $\mu$ be an ergodic Borel probability measure on the hull $X_D$. 
	Assume \Az and \Ao and let $\nu_{D}$ be the integrated density of states \eqref{intIDS}. 
	Let $L>0$ and let $E\in \R$ be a continuity point of $\nu_D$.
	Then the integrated density of states satisfies Dirichlet--Neumann bracketing averaged over the hull
	\be
		\label{ndbrack} 
		\int_{X_D} \mathbb E_{P}\left(\tilde\nu_{\pom\!,L}^{\rm D}(E)\right) \d\mu(P) 
		\leq \nu_D(E)\leq \int_{X_D} \mathbb E_{P} \left(\tilde\nu_{\pom\!,L}^{\rm N}(E)\right)\d\mu(P).
	\ee
\end{lem}

\begin{rem}
	\label{ndbrack-cont}
  If \Aloc is assumed in addition, then \eqref{ndbrack} holds for every $E\in\mathbb R$. This follows from Remark \ref{nu_cont}.
\end{rem}


\begin{proof}[Proof of Lemma~\ref{ndbracklem}]
For a fixed $L>0$ we consider the sequence $\{\L_{K}\}_{K\in L\mathbb N}$ of concentric open cubes in $\R^{d}$ centered about the origin. For every $K\in L\N$ the decomposition
\be\label{cubedecoup} 
	\L_{K} = \Bigg(\bigcup_{j\in\mathcal J}\overline{\L_{L}(j)}\Bigg)^{\text{int}}
\ee
holds with some index set $\mathcal J \subset \Rd$ of cardinality $|\mathcal J|=(K/L)^d$.
By the subadditivity of $\tr \,\chi_{]-\infty,E]}(H_{\pom\!,K}^{\rm N})$ in the volume we have
\begin{align}
	\label{neumann-decomp}
	\tilde\nu_{\pom\!,K}^{{\rm N}}(E)= \frac{1}{K^d} \tr\, \chi_{]-\infty,E]}(H_{\pom\!,K}^{\rm N})
		&\le \frac{L^d}{K^d} \displaystyle\sum_{j\in\mathcal J} \frac{1}{L^d} \tr \, 
		\chi_{]-\infty,E]}(H_{\pom\!,L,j}^{\rm N})\nonumber\\
	&= \frac{1}{|\mathcal J|} \displaystyle\sum_{j\in\mathcal J} \tilde\nu_{\pom\!,L,j}^{\rm N}(E)
\end{align}
for every $E\in\mathbb R$ and every $P^{\omega} \in \hat{X}_{D}$.
Integrating with respect to the translation-invariant measure $\hat\mu$ on $\hat X_D$, we obtain
\begin{equation}
	\label{avDNb}
	\int_{\hat X_D} \!\tilde \nu_{\pom\!,K}^{\rm N}(E) \,\d\hat\mu(\pom)
	 \leq \int_{\hat X_D}\! \tilde\nu_{\pom\!,L}^{\rm N}(E) \,\d\hat\mu(\pom).
\end{equation}
Next, we take the limit $K \rightarrow \infty$, keeping $L$ fixed. Since $H_{\pom\!,K}^{\rm N}\geq-\Delta^{\rm N}_{K}$, it follows that $\tilde \nu_{\pom\!,K}^{\rm N}$ is bounded from above by a quantity that does not depend on $\pom$ so that Lebesgue's
Dominated Convergence Theorem and Proposition \ref{approxnu} yield
\be\label{conv_nu} 
	\lim_{K\rightarrow \infty} \int_{\hat X_D} \tilde \nu_{\pom\!,K}^{\rm N} (E) \,\d\hat\mu(\pom) 
	= \int_{\hat X_D}\nu_D(E) \,\d\hat\mu(\pom)=\nu_D(E)
\ee
for every continuity point $E$ of $\nu_D$. Thus, \eqref{conv_nu}, \eqref{avDNb} and  \eqref{erg1} 
provide the upper bound in \eqref{ndbrack}.

In an analogous way, using the superadditivity of $\tr\,\chi_{]-\infty,E]}(H_{\pom\!,K}^{\rm D})$ in the volume, we obtain the lower bound in \eqref{ndbrack}.
\end{proof}

Having established Dirichlet--Neumann bracketing, we need upper and lower Lifshitz-tail bounds 
for the integrands in \eqref{ndbrack} to prove Theorem \ref{lt}. This is done in Theorem~\ref{unifbd}. The proof of Theorem~\ref{unifbd} follows the principal ideas which are well known for impurities located on periodic point sets. However, we need  additional efforts to make the result useful for proving localization for \emph{every} Delone set in the hull $X_{D}$, in particular for $D$ itself.  Usually, the length $L \sim E^{-1/2}$ is determined and fixed by the energy in such estimates. But in contrast to the periodic case we cannot work with the infinite-volume integrated density of states to establish the initial estimate, because self-averaging of the finite-volume integrated density of states in the macroscopic limit is only known for almost every Delone set in the hull. Nor do we know monotonicity in $L$ of the finite-volume Dirichlet or Neumann integrated densities of states, which holds pointwise for every Delone set in the hull $X_{D}$. We know this only in average over the hull. The way out will 
be to extend the Lifshitz-tail estimates to all sufficiently large lengths $L \gtrsim E^{-1/2}$. It is also important that the constants in the estimates are uniform on the hull $X_{D}$. 

\begin{thm}\label{unifbd}
	Let $D$ be a Delone set with radii of uniform discreteness $r$ and relative denseness $R$. 
	Assume \Az with $B=0$, \Ao and \Aloc. Then 
	\begin{nummer}
 	\item
		for every $E' >0$ there exist constants $C^{(1)}, C^{(2)}, C^{(3)}>0$, 
		depending all on $r$, $u$ and $\rho$, and a constant $C_{d}>0$ such that 
		\be
 			\mathbb E_{P}\left(\tilde\nu_{\pom\!,L}^{\rm N}(E)\right)  \leq C_{d} \e^{-C_R E^{-d/2}}
		\ee
		for every $P \in X_{D}$, every $E \in ]0,\min\{E_{R},E'\}[$ and every $L\in [{L}_{E,R}, \infty[$, where 
		$E_{R}:= C^{(1)} R^{-d}$, $L_{E,R} := C^{(2)} R^{-d/2} E^{-1/2}$ and 
		\begin{equation}
 			\label{CR} 
			C_{R} := C^{(3)} R^{-d-d^{2}/2}.
		\end{equation}
	\item	
		if \Alt holds in addition, then there exists constants $c_{d}^{(1)}, c_{d}^{(2)}, C_{r,\rho}$, 
		$C_{u, \rho}, E_{u,\rho}, >0$ 
		such that 
		\be  
			\mathbb E_{P}\left( \tilde\nu_{\pom\!,L}^{\rm D}(E)\right)
			\geq c_{d}^{(1)} E^{d/2} \e^{-C_{r,\rho} E^{-d/2}\abs{\ln C_{u,\rho} E}}
		\ee
		for every $P \in X_{D}$, every $E \in ]0,E_{u,\rho}[$ and every $L\in [\ell_{E}, \infty[$, where 
		$\ell_{E} := c_{d}^{(2)} E^{-1/2}$.
	\end{nummer}
\end{thm}

\begin{proof}
	{(i)} \quad 
	Let $P\in X_{D}$, $L>0$ and $0< E \le E'$. To obtain the upper bound we observe that
	\begin{equation}
		\label{ub-0}
 		\mathbb E_{P} \big(\tilde\nu_{\pom\!,L}^{\rm N}(E)\big) \le C_{d} \, 
		\mathbb P_P\big( E_1( H_{\pom\!,L}^{\rm N}) \le E  \big),
	\end{equation}
	where $E_1( H_{\pom\!,L}^{\rm N})$ denotes the ground-state energy of $H_{\pom\!,L}^{\rm N}$,
	and the constant $C_{d}$ can be chosen as an $L$-independent upper bound on the integrated 
	density of states of the Neumann Laplacian on $\Lambda_{L}$ at energy $E'$.
	Next, Temple's inequality provides the second estimate in
	\begin{equation}
		\label{ub-1}
 		E_1( H_{\pom\!,L}^{\rm N}) \ge E_1( H_{P^{\tilde{\omega}}\!,L}^{\rm N}) 
		\ge \frac{c_{u}}{L^{d}} \sum_{p\in \Lambda_{L} \cap P }\tilde{\omega}_{p},
	\end{equation}
	where $\tilde\omega := (\tilde\omega_{p})_{p\in P}:= \big(\min\{\omega_p, \alpha L^{-2}\}\big)_{p \in P}$ 
	defines the couplings of a suitably truncated random potential. This estimate is valid if $L> \delta_{u}$. 
	The constants can be chosen as $\alpha := [u^{+}(2\delta_{u}/r)^{d}]^{-1}$ and 
	$c_{u}:= 2^{-d-1} u^{-} \varepsilon_{u}^{d}$. They are uniform in $\pom\in X_{D}$ and do not depend $R$.
	To derive this estimate in \eqref{ub-1}, we have adapted the 
	argument in \cite[Sect.\ 6.2]{K} to our continuum setting, using 
	\be
		\label{delone}
		\frac{L^d}{R^d} \leq|\L_L \cap P|\leq \frac{L^d}{r^d}
		\ee
	and the assumptions on the single-site potential in \Aloc.
	
	We conclude from \eqref{ub-0}, \eqref{ub-1} and the Large-Deviation Lemma~\ref{LargeDev} that
	\begin{equation}
 		\mathbb E_{P} \big(\tilde\nu_{\pom\!,L}^{\rm N}(E)\big) 
		\le C_{d} \, \mathbb P_P\Bigg( \frac{c_{u}}{L^{d}} 
			\sum_{p\in \Lambda_{L} \cap P }\tilde{\omega}_{p} \le E \Bigg)
		\le C_{d} \e^{-C_0 R^{-d} L^{d}}
	\end{equation}
	for every $E \in ]0,\min\{4E_{R},E'\}[$ and every $L\in [\ell, 2L_{E}]$, where 
	$\ell := \max\{L_{0}, \delta_{u}\}$, $E_{R} := c_{u}\alpha (4\ell)^{-2} R^{-d}$ and
	$L_{E} := (1/4) (\alpha c_{u}/E)^{1/2} R^{-d/2}$. The constants $C_{0}$ and $L_{0}$ are those 
	from Lemma~\ref{LargeDev} and do not depend on $R$. 
	We note that $E< E_{R}$ implies $L_{E} \ge \ell$. Therefore we have
	\begin{equation}
		\label{ub-ul}
 		\mathbb E_{P} \big(\tilde\nu_{\pom\!,L}^{\rm N}(E)\big) 
		\le C_{d} \e^{-C_R E^{-d/2}}
	\end{equation}
	for every $E \in ]0,\min\{E_{R},E'\}[$, every $L\in [{L}_{E}, 2L_{E}]$ and every $P\in X_{D}$ with
	$C_R := C_0 (\alpha c_{u}/16)^{d/2} R^{-d-d^2/2}$.
	
	Finally, we extend \eqref{ub-ul} to arbitrarily large lengths. Let $L' > 2L_{E}$, then the integer part 
	$k := \lfloor L'/L_{E} \rfloor \in \N$ and $L := L'/k \in [L_{E}, 2L_{E}]$. We use the subadditive 
	Neumann decomposition \eqref{neumann-decomp} with $K=L'$ and $|\mathcal{J}|=k^{d}$. 
	Since \eqref{ub-ul} holds for every $P\in X_{D}$, in particular for every shifted point set $j+P$, 
	$j\in \R^{d}$, with the same constants, we arrive at 
	\begin{equation}
		\mathbb E_{P} \big(\tilde\nu_{\pom\!,L'}^{{\rm N}}(E) \big) 
		\le \frac{1}{k^{d}} \sum_{j\in\mathcal J} \mathbb E_{P} \big( \tilde\nu_{\pom\!,L,j}^{\rm N}(E) \big) 
		\le C_{d} \e^{-C_R E^{-d/2}}
	\end{equation}
	for every $E \in ]0,\min\{E_{R},E'\}[$, every $L'\in [2{L}_{E}, \infty[$ and every $P \in X_{D}$. Together with \eqref{ub-ul}, this proves the claim. 

	{(ii)} \quad
	To verify the lower bound
	we let $E,L >0$ and follow the same strategy as for the standard alloy-type model. 
	We estimate the Dirichlet ground-state energy by the min-max principle
	\begin{align}
		\label{lb-start}
		\mathbb E_{P}\left( \tilde\nu_{\pom\!,L}^{\rm D}(E)\right) 
		& \geq \frac{1}{L^d} \,\mathbb P_P\big( E_1( H_{\pom\!,L}^{\rm D}) \le E  \big) \nonumber\\
		& \ge  \frac{1}{L^d} \,\mathbb P_P\Big( \langle\psi, H_{\pom\!,L}^{\rm D} \psi\rangle \le E \|\psi\|^{2} \Big),
	\end{align}
	where $\psi(x) := L^{-d/2} \chi(x/L)$ for $x\in \Lambda_{L}$ and $\chi \in C_{c}^{\infty}(\Lambda_{1})$ 
	is a cut-off function such that 
	$0 \le \chi \le 1$, $\chi|_{\Lambda_{1/2}}=1$ and $\supp\chi \subseteq \Lambda_{3/8}$.

	Now, we assume that $L \ge 4\delta_{u}$. Then, the condition on the support of $\chi$ implies that the 
	expectation of the potential energy is not influenced by impurities outside of $\Lambda_{L}$, and we obtain
	$\langle\psi, V_{\pom}\psi\rangle \le \frac{\|u\|_{1}}{L^{d}} \, \sum_{p\in P\cap \Lambda_{L}} \omega_{p}$.
	The norm of $\psi$ is bounded from below, $\|\psi\|^{2} = \int_{\Lambda_{1}} \chi(x) \d x \ge 2^{-d}$
	and the expectation of the kinetic energy is given by	$\langle\psi, -\Delta_{L}^{\rm D} \psi\rangle = L^{-2} \int_{\Lambda_{1}} |(\nabla\chi)(x)|^{2} =: c L^{-2}$.
	From this we infer that
	\begin{equation}
		\mathbb E_{P}\left( \tilde\nu_{\pom\!,L}^{\rm D}(E)\right) 
		\ge  \frac{1}{L^d} \, \mathbb P_P\Bigg( \frac{1}{L^{d}} \sum_{p\in P\cap \Lambda_{L}} \omega_{p}  
				\le \frac{2^{-d}E -cL^{-2}}{\|u\|_{1}} \Bigg)
	\end{equation}
	for every $L\ge 4\delta_{u}$. From here on we assume 
	\begin{align}
 		L &\ge \ell_{E} := (2^{d+1}c/E)^{1/2} \qquad \quad \text{and}  \label{L-cond}\\
		E &\le 2^{d-3} c/\delta_{u}^{2}.  \label{E-cond}
	\end{align}
	The condition on $L$ implies that $2^{-d}E -cL^{-2} \ge 2^{-d-1}E$	and the one on $E$ that 
	$\ell_{E} \ge 4\delta_{u}$. Thus, $L \ge 4\delta_{u}$ holds, too. Assuming \eqref{L-cond} 
	and \eqref{E-cond} and using \Alt, we get 
	\begin{align}
		\mathbb E_{P}\left( \tilde\nu_{\pom\!,L}^{\rm D}(E)\right) 
		& \ge  \frac{1}{L^d} \, \mathbb P_P\Bigg( \frac{1}{L^{d}} \sum_{p\in P\cap \Lambda_{L}} \omega_{p}  
				\le \frac{E}{2^{d+1}\|u\|_{1}} \Bigg) \nonumber\\
		& \ge  \frac{1}{L^d} \, \mathbb P_P\Bigg( \forall\, p\in P \cap \Lambda_{L}: \, \omega_{p}  
				\le \frac{E}{2^{d+1}\|u\|_{1}} \Bigg) \nonumber\\
		& \ge  \frac{1}{L^d} \,	\bigg[ C_{\rho} \bigg(\frac{E}{2^{d+1}\|u\|_{1}}\bigg)^{\beta} 
				\bigg]^{|P \cap \Lambda_{L}|}
	\end{align}
	Now, we replace \eqref{L-cond} and \eqref{E-cond} by the stronger hypotheses	
		\begin{align}
 		&L \in [\ell_{E}, 2\ell_{E}] \qquad\quad \text{and}  \\
		&E \le E_{u,\rho} := \min \big\{ 2^{d-3} c/\delta_{u}^{2}, C_{\rho}^{-1/\beta} 2^{d+1} \|u\|_{1} \big\}   
	\end{align}
	so that 
	\begin{equation}
		\label{lb-ul}
		\mathbb E_{P}\left( \tilde\nu_{\pom\!,L}^{\rm D}(E)\right) 
		\ge  \frac{1}{(2\ell_{E})^d} \,	\bigg[ C_{\rho} \bigg(\frac{E}{2^{d+1}\|u\|_{1}}\bigg)^{\beta} 
				\bigg]^{(2\ell_{E}/r)^{d}}
	\end{equation}
	where we also used \eqref{delone}. 

	Finally, we need to extend \eqref{lb-ul} to arbitrarily large lengths. 
	This we do in the same way as we did for the upper bound in (i), but with the superadditive 
	Dirichlet decomposition replacing \eqref{neumann-decomp}. This gives the claim.	
\end{proof}

The following large-deviation principle has been used in the previous proof. It is adapted from \cite[Lemma 6.4]{K} with some constants made explicit.

\begin{lem}\label{LargeDev} 
	Let $P\in X_{D}$ and let $(\omega_p)_{p\in P}$  non-negative i.i.d.\ random variables 
	whose single-site distribution $\mathbb{P}^{(0)} $has no atoms. 
	Given $\alpha,L>0$, let $\tilde\omega_p:=\min\{\omega_p,\alpha L^{-2}\}$ for $p\in P$. 
	Then, there exists $L_{0}>0$, which depends on $\alpha$ and $\mathbb{P}^{(0)}$, and there exists $C_{0}>0$, 
	which depends only on  $\mathbb{P}^{(0)}$, such that for every $E \in ]0, \alpha (2L_{0})^{-2}R^{-d}[$ 
	and every $L \in [L_{0}, (1/2)(ER^d/\alpha)^{-1/2}]$ the large-deviation estimate 
	\be  
		\mathbb P_P \Bigg( \frac{1}{L^d} \sum_{p\in\L_L \cap P} \tilde\omega_p \leq E \Bigg)  
		\leq \e^{-C_0\frac{L^d}{R^d}}
	\ee
	holds.
\end{lem}

\begin{proof}
	The number of points in $\L_L \cap P$ is not necessarily equal to $L^d$, so we cannot apply 
	directly a large-deviation principle. However, \eqref{delone} implies
	\be
		\label{a:toLDP} 
		\mathbb P_P \Bigg( \frac{1}{L^d} \sum_{p \in\L_L \cap P} \tilde\omega_p \leq E \Bigg) 
		\leq \mathbb P_P \Bigg( \frac{1}{|\L_L \cap P|}\sum_{p \in\L_L \cap P} \tilde\omega_p \leq \tilde E \Bigg), 
	\ee
	where $\tilde E=E\,R^d$. The r.h.s.\ can be estimated by the large-deviation estimate \cite[Lemma 6.4]{K}, 
	which, upon inspection of the constants, using that the single-site distribution has no atoms 
	and applying \eqref{delone}, yields the claim. 
\end{proof}

\begin{proof}[Proof of Theorem \ref{lt}]
	Using the notation of Theorem~\ref{unifbd}, let $E \in ]0, \min\{E_{R}, E', E_{u,\rho}\}[$ and 
	choose some $L \ge \max\{L_{E,R}, \ell_{E}\}$. Since all the constants in Theorem~\ref{unifbd} are 
	uniform in $P \in X_{D}$ and since $\mu$ is a probability measure on $X_{D}$, 
	Lemma~\ref{ndbracklem}, Remark~\ref{ndbrack-cont} and the bounds from 
	Theorem~\ref{unifbd} imply 
	\begin{align}
		\label{uppIDS}
		c_{d}^{(1)} E^{d/2} \e^{-C_{r,\rho} E^{-d/2}\abs{\ln C_{u,\rho} E}} 
  	\le \nu_D(E)
		\leq C_{d}\e^{-C_R E^{-d/2}}
	\end{align}
	for every $E \in ]0, \min\{E_{R}, E', E_{u,\rho}\}[$. The claim now follows. 
\end{proof}

\subsection{Application to dynamical localization}\label{s:dynloc}

In this subsection we describe how to use the bounds from Theorem~\ref{unifbd} as an ingredient of the bootstrap multi-scale analysis (MSA) to obtain dynamical localization at the bottom of the spectrum. For this purpose it is important that these bounds are uniform on the hull $X_D$. 
We emphasise that we cannot work with the infinite-volume integrated density of states to establish the initial estimate, because self-averaging of the finite-volume integrated densities of states in the macroscopic limit is only known for almost every point set in the Delone hull, and thus not necessarily for $D$ itself.
We also obtain a lower bound on the size of the region of dynamical localization in terms of the radius $R$ of relative denseness of the Delone set $D$.

%
%

\begin{thm}\label{t:dynlocas}
Let $D$ be a Delone set and assume \Az with $B=0$, \Ao and \Aloc. Then, there exists an energy $E_{LT}(R)$ such that the operator $H_{P^{\omega}}$ exhibits dynamical localization in $[0,E_{LT}(R)]$ for all $P\in X_D$ and $\mathbb{P}_{P}$-a.a.\ $\omega\in \Omega_P$. Moreover, there exist a constant $R_0>1$ and constants $C_{1},C_{2} >0$ such that for every $R\geq R_0$ we have
  \be \label{f:ELT}
 E_{LT}(R)=C_1 R^{-\left(d+2+\frac{8}{3d}\right)} (\log C_2 R)^{-\frac{2}{d}}.\ee
The constants $C_1$ and $C_2$ depend on the parameters of the model, but not on the radius $R$ of relative denseness of $D$.
\end{thm}

The proof of this theorem is given below Theorem~\ref{t:finvolcrit}.
For large values of the parameter $R$, the use of Theorem~\ref{unifbd} gives a better lower bound for the interval of dynamical localization than the previous approach based on the space-averaging approximation from \cite{G, BoK, RM13}.
We can see this by comparing Theorem \ref{t:dynlocas} to

\begin{thm}\label{t:dynloc} 
Let $D$ be a Delone set and assume \Az, with $B=0$, \Ao and \Aloc. Then, there exists an energy $E_{SA}(R)$ such that the operator $H_{P^{\omega}}$ exhibits dynamical localization in $[0,E_{SA}(R)]$ for all $P\in X_D$ and $\mathbb{P}_{P}$-a.a.\ $\omega\in \Omega_P$. Moreover, there exist a constant $R_0'>1$ and constants $C_{1}', C_{2}' >0$ such that for every $R\geq R_0'$ we have
\be \label{f:ESA}
E_{SA}(R)=C_1' R^{-(4d+4)}(\log C_2' R)^{-\frac{2}{d}}.\ee
The constants $C_1'$ and $C_2'$ depend on the parameters of the model, but not on the radius $R$ of relative denseness of $D$.
\end{thm}

\begin{rem} \begin{nummer} 
\item The proof of this theorem is given at the end of this section.
\item There exists $R_*:=\max\{R_0,R_0'\}>0$ depending on $C_1,C_2,C_1',C_2'$, such that for $R\geq R_*$, Theorem \ref{t:dynlocas} gives a better lower bound for the region of dynamical localization than Theorem \ref{t:dynloc}.
\item For an analogous discrete model, \cite{EKl} obtained a spectral gap at the bottom of the spectrum of order $R^{-2d}$. In the continuous setting, such an estimate would give a lower bound for the region of dynamical localization of order $R^{-(2d+\frac{8}{3d})}\left(\log C R\right)^{-\frac{2}{d}}$ in Theorem \ref{t:dynloc}.
\end{nummer}
\end{rem}

Under our assumptions on the random potential, we can apply the MSA from \cite{GK1} to prove dynamical localization. To do so, it is enough to verify the two main ingredients concerning the finite-volume operator $H_{D^\omega\!,L,y}$: the \emph{Wegner estimate} and the \emph{initial-length-scale estimate}. The MSA method generalizes to Delone--Anderson models by requiring these estimates to hold for $H_{D^\omega\!,L,y}$ uniformly  with respect to $y\in\Rd$ \cite[Theorem 2.3]{RM}. Uniform Wegner estimates have been obtained for our model both in the continuous \cite{RMV,Kl} and in the discrete setting \cite{EKl,RM13}.


In order to obtain a lower bound on the region of dynamical localization, we will use a finite-volume criterion from \cite{GK2}  that gives a sufficient condition on the length scale $L$ in order to start the MSA.
We recall that our model satisfies the usual structural conditions required to apply the MSA uniformly with respect to the center of the box (see \cite{RM, GK1}): the Simon-Lieb inequality (SLI), the eigenfunction decay inequality (EDI),
the strong generalized eigenfunction expansion (SGEE) and the average number of eigenvalues (NE). 
The constant $\gamma_I$ in (SLI) is uniform on subsets $I \subseteq [0,1]$ and on the hull $X_{D}$.
 
 Note that the condition of independence of events at a distance $\varrho>0$ (IAD) is ensured by having single-site potentials with compact support. For a fixed energy $E$, we denote by $\eta_I$ the maximal length of an interval $I$ containing $E$ for which the Wegner estimate holds with a constant $Q_I$.  We write $\Gamma_{L}:=\chi_{\bar{\Lambda}_{L-1}\setminus\Lambda_{L-3}}$ and $\chi_{L/3}:=\chi_{\Lambda_{L/3}}$.

\begin{thm}[\protect{\cite[Theorem 2.5]{GK2}}]\label{t:finvolcrit}
Let $(H_\omega)_{\omega\in\Omega}$ be a random operator such that Assumptions SLI, EDI, IAD, NE, SGEE and a Wegner estimate with volume dependence $L^{bd}$, $b\ge 1$, hold in an open interval I.  Given $s>bd$, we set
\be\label{f:finvol} \mathcal L := \max \left\{ 3\varrho, 42, 3\left(\frac{107^d}{2}\right)^{\frac{2}{s-bd}}, \frac{1}{37}(16\cdot 60^dQ_I)^{\frac{2}{s-bd}},\eta_I^{-\frac{1}{s}} \right\}. \ee
Suppose that for some $L\geq \mathcal L$, $L\in 6\mathbb N$, and some $E_0\in \Sigma\cap I$, where $\Sigma$ is the non-random spectrum of $(H_\omega)_{\omega\in\Omega}$, we have the upper bound
\be\label{ILSEprob}\Prob{ 90^d\gamma_I^2(37L)^{s} \norm{\Gamma_{L} (H_{\om,L}-E_{0})^{-1}\chi_{L/3}} <1}\geq 1-\frac{2}{344^d}. \ee
Then, $E_0$ is in the region of dynamical localization.
\end{thm}

\begin{proof}[Proof of Theorem \ref{t:dynlocas}]
For $L \in 6 \N$ we define
\be\label{f:EILSE} E_*(L):=\frac{1}{2} \left( \frac{C_R}{(p+2)d\ln L} \right)^{\frac{2}{d}}, \ee
where $C_R:=C^{(3)} R^{-(d+d^2/2)}$ was introduced in \eqref{CR}. 
We use the Wegner estimate from \cite[Theorem 1.5]{Kl} applied to Delone sets, see \cite[Remark 1.6]{Kl}: given $E_0>0$, there exist positive constants $C_W=C_W(d,u,w,\rho,E_0)$, $c=c(d)$ and $\gamma_R$ to be made explicit later,
such that for any closed interval $I\subset ]-\infty,E_0]$ with $\abs{I}\leq 2\gamma_R$ and $L\geq 72\sqrt{d}+\delta_u$, we have
\be
	\label{eq:WE} 
	\E_P\left( \tr \chi_I(H_{\pom\!,L,y}^{\rm D}) \right)
	\le \E_P\left( \tr \chi_I(H_{\pom\!,L+1,y}^{\rm D}) \right)
	\leq C_W \gamma_R^{-2c} \abs{I}L^d . 
\ee
Here, the intermediate step with a box of side length $L+1$ is included because \cite[Theorem 1.5]{Kl} requires $L$ to be odd. 

We will apply Theorem \ref{t:finvolcrit}, with $s=2d$ and $b=1$, to prove that $[0,E_*(L)]$ belongs to the region of complete localization. It remains to determine $L$ such that \eqref{f:finvol} and \eqref{ILSEprob} hold. Concerning the initial estimate \eqref{ILSEprob}, we consider the event
\be\label{event} \mathcal E_{P,L}:=\{\omega\in\Omega_P\,:\, H_{\pom\!,L,y}^{\rm D}\geq 2E_*(L)\} \ee
of having a spectral gap. Then, the Combes--Thomas estimate \cite[Theorem 2.4.1]{S} implies that there exist positive constants $C_1=C_1(w,u,r)$ and $C_2=C_2(w,u,r)$ such that for any $E\in [0,E_*(L)]\subset [0,1]$,
\be \norm{\Gamma_{L}(H_{\pom,L,y}^{\rm D} - E)^{-1}\chi_{L/3}}\leq \frac{C_1}{E_*(L)}\e^{-C_2 \sqrt{E_*(L)}L}.   \ee
Therefore, given $\alpha>1+\frac{1}{d}$, there exist positive constants $C$ and $R_1$ depending on the parameters $\alpha,d,p,C_1,C_2,C^{(3)},\gamma_I$ such that the event on the l.h.s. of \eqref{ILSEprob} holds for every $\omega\in\mathcal E_{P,L}$ provided
\be\label{a1} L\geq C_{}R^{\,\alpha(\frac{d}{2}+1)}  \ee
and the radius of relative denseness satisfies $R>R_1$. As for the probability of this event, using the Chebyshev--Markov inequality and the fact that the Neumann restriction $H_{\pom\!,L}^{\rm N}$ is dominated by the Dirichlet restriction, we have that
\be \P_P\left( \mathcal E_{P,L}^c\right) \leq  L^d\mathbb E_{P}\left(\tilde\nu_{\pom\!,L}^{\rm N}(E)\right). \ee
In order to use Theorem~\ref{unifbd} to bound the term in the r.h.s., we impose 
\be E_*(L)\leq \frac{C^{(1)}}{R^d} \quad \mbox{and} \quad L> \frac{C^{(2)}}{R^{d/2}E_*(L)^{1/2}}. \ee
The first condition is fullfilled if \eqref{a1} holds with $R>R_3:=\max\{R_1,R_2\}$, where $R_2=R_2(p,c, C^{(1)},C^{(3)})$.
As for the second condition, it is enough to have
\be\label{a1prime} L\geq C'R^{\alpha'}, \ee
where $\alpha'>2d$ and $C'=C'(d,p,C^{(2)},C^{(3)})$. We can now use Theorem~\ref{unifbd} and conclude
\begin{align}\label{probailse}
\P_P\left( \mathcal E_{P,L}^c\right) &\leq C_d L^d \e^{-C_R(2E_*(L))^{-d/2}}\nonumber\\
& \leq C_dL^{-d(p+1)}.
\end{align}
Therefore, the initial estimate \eqref{ILSEprob} holds provided 
\be\label{a2prime} L\geq C_0R^{\alpha_0},\ee
with  $R>R_3$, where $C_0:=\max\{C,C'\}$, $\alpha_0:=\max\{\alpha,\alpha'\}$ and $R_3$ depends on $\alpha_0,d,p,C_1,C_2, C^{(1)},C^{(2)}, C^{(3)},\gamma_I$. 

It remains to verify \eqref{f:finvol} for $L$. For this we will determine the constants $Q_I^{2/d}$ and $\eta_I^{-1/2d}$ in terms of $R$.
From the Wegner estimate in \eqref{eq:WE} we see that the constant $Q_I$ is of the form $C_W \gamma_R^{-2c}$, where $C_W=C_W(d,u,w,\rho)$, $c=c(d)$ and $\gamma_R$ is a constant that depends on $R$.
This is the constant appearing in the positivity estimate \cite[Eq. 1.9]{Kl} and comes from the unique continuation principle \cite[Theorem 2.1]{Kl}. The dependence on $R$ of the unique continuation principle was made explicit in 
\cite[Corollary 2.2]{RMV}, which gives
\be \gamma_R =c_{u,d}R^{-c_{d,u,r,w}(R^{4/3}+c_d)}. \ee
Then, we have $Q_I^{2/d}= C'_W R^{c_{d,u,r,w}(R^{4/3}+c_d)}$, $\eta_I^{-1/2d}=c'_{u,d}R^{c'_{d,u,r,w}(R^{4/3}+c_d)}$, where $C'_W=\tilde C'_W(d,u,w,\rho)$.
 Therefore, recalling \eqref{a2prime} and \eqref{f:finvol}, it is enough to have
\be\label{a2} L\geq \tilde C R^{\tilde \alpha R^{4/3}}=:L_R,   \ee
with $R>R_0$, where $\tilde\alpha=\tilde \alpha(d,u,r,w,\alpha_0)>1$, and $\tilde C$ and $R_{0}$ depend on the parameters of the model except for $R$.
By Theorem \ref{t:finvolcrit} we conclude that $[0,E_{LT}(R)]$ is in the region of dynamical localization, where
\be\label{Elt}  E_{LT}(R):=E_*(L_R)=C_1 R^{-\left(d+2+\frac{8}{3d}\right)}(\log C_2\,R)^{-\frac{2}{d}}, \ee
with $R\geq R_0$, where $C_2=\tilde C$ and the constant $C_1$ depends on $p,d,C^{(3)},\tilde\alpha$.

Note that the bounds only depend on the parameters $r,R$ and are uniform on $X_D$, therefore, they apply to any $P\in X_D$ and, in particular to $D$.
\end{proof}

\begin{proof}[Proof of Theorem \ref{t:dynloc}]
The initial length scale estimate has been obtained previously for Delone-Anderson models by \cite{G} using a space-averaging argument found in \cite{BoK}, which can also be used to obtain the Wegner estimate, see \cite{RM13}.
One can use the output from \cite{G, RM13} to obtain an analogue of \eqref{f:EILSE} in the proof of Theorem \ref{t:dynlocas}.
Taking a more precise account of the role of the parameter $R$ in the large-deviation estimate used in \cite{G, RM13}, one needs to rewrite \cite[Eq.\ (3.10)]{G} and \cite[Eq.\ (2.20)]{RM13}.
First, the lower bound for the space-averaged version of the finite-volume random potential in \cite[Eq.\ (3.8)]{G}
needs to be written as
\begin{align} \bar V_{\dom\!,L} & :=\frac{1}{(2RK)^d}\int_{\L_{2RK}} V_{\dom\!,L}(\cdot-a)\,\d a \nonumber\\
&\geq \frac{c_{u,d}}{(2R)^d(3R)^d}\left(\min_{\xi_\in \Lambda_L\cap D}\frac{(3R)^d}{K^d}\sum_{p\in \Lambda_{\frac{K}{3}}(\xi)\cap D}\omega_p\right) \chi_{\Lambda_L},
\end{align}
where $K$ is a large constant to be chose later. Then one can use \eqref{delone} as in \eqref{a:toLDP} and a large-deviation estimate to see that, for arbitrary fixed $\xi$, \cite[Eq.\ (3.10)]{G} becomes
\be \mathbb P_P \left( \frac{(3R)^d}{K^d} \sum_{p \in \Lambda_{\frac{K}{3}}(\xi)\cap D} \omega_p \leq \frac{\bar{\rho}}{2} \right) \leq \e^{-A_\rho \frac{K^d}{R^d}},
\ee
where now the constant $A_\rho>0$ depends only on the probability distribution of the random variables, and $\bar\rho=\mathbb E(\omega_0)$. This yields for \cite[Eq.\ (3.11)]{G}:
\be  \mathbb P_P \left( \bar V_{\dom\!,L} >c_{u,d,\rho}'R^{-2d} \chi_{\Lambda_L}\right)\geq 1-\frac{L^d}{r^d}\e^{-A_\rho \frac{K^d}{R^d}}. \ee
Following the rest of the proof of \cite[Prop.\ 3.1]{G} one obtains the existence of positive constants $C_{u,d}$ and $K_{u,d}$ such that $H_{\pom,L,y}^{\rm D} \geq C_{u,d,\rho} R^{-(4d+2)}K^{-2}$
 with a probability larger than $1-r^{-d}L^d \e^{-A_\rho \frac{K^d}{R^d}}$ for $K>K_{u,d}$. Choosing
\be\nonumber K=R\left( \frac{(p+2)d\log L}{A_\rho}\right)^{\frac{1}{d}}, \ee
and taking $R>K_{u,d}$, implies that the event $\mathcal E_{P,L}$ defined in \eqref{event}, with
\be E_*(L)=C_{u,d,\rho,p}R^{-(4d+4)}(\log L)^{-\frac{2}{d}},\ee
has  probability $\P_P(\mathcal E_{P,L})> 1-L^{-pd}$, for $L>L_{r,d}$, where $L_{r,d}$ is a constant depending on $r$ and $d$. 

The lower bound for the region of dynamical localization can be obtained by applying Theorem \ref{t:finvolcrit} with $s=2d$ and $b=1$, in the same way as in the proof of Theorem \ref{t:dynlocas}.
We will take the Wegner estimate from \cite[Theorem 2.1]{RM13} obtained by the space-averaging approach, see \cite[Theorem 4.2.1]{RMthesis} for the continuous setting.  This states that there exist positive constants
$C_{u,d}$ and $C_W=C_W(u,d,\rho)$ such that 
\be\label{eq:WE2} \E_P\left( \tr \chi_I(H_{\pom\!,L,y}^{\rm D}) \right)\leq C_W R^{2d} \abs{I}L^d, \ee
for any $I\subset [0,E_W[$, where $E_W:=C_{u,d} R^{-2(d+1)}$.
Concerning condition \eqref{ILSEprob}, it is enough to have, for a given   $\theta>1+1/d$,
\be\label{L(R)} L\geq C_{\theta,d,u,w,\rho}R^{2\theta(d+1)},\ee
with $R>R_1'$, where $R_1'=R_1'(u,d,p,\rho,w,\theta)$. For the remaining condition \eqref{f:finvol},
note that in order to apply the Wegner estimate \eqref{eq:WE2}, we need to have  $[0,E_*(L)]\subset [0,E_W)$. For this it is enough to take $L$ as in \eqref{L(R)} with $R\geq R_2'$, where $R_2'=R_2'(u,d,\rho,p,r)$.
Note that from \eqref{eq:WE2} we have $Q_I^{2/d}=c_{u,d,\rho} R^{4}$ and $\eta_I^{-1/2d}=c_{u,d}R^{(1+1/d)}$.
Then, to have \eqref{f:finvol}, it is enough to ask $L\geq C_{\theta,d,u,w,\rho}'R^{\theta}$, with $\theta>1+1/d$, $R>R'_0:=\max\{ R_1',R_2'\}$, and we can proceed as in Eq. \eqref{a2},
\eqref{Elt}.
\end{proof}


\appendix

\section{An example of a non uniquely ergodic Delone set}
\label{sec:example}

The purpose of this appendix is to demonstrate that the assumption of \emph{unique} ergodicity in Corollary~\ref{eids}(ii) cannot simply be dropped.


Consider a sequence of (open) cubes centered about the origin, $\{\L_{L_k}\}_{k\in\N}$, with $L_{k+1}=L_k^\alpha$, $\alpha>1$.  Define $\N_e=\{2k:\, k\in\N\}$, $\N_o=\{2k-1:\, k\in\N\}$, and consider the following covering of $\Rd$,

\be\label{cov} \Rd= \bigcup_{k=1}^\infty A_k, \quad A_k:=\overline{\L}_{L_k}\setminus \L_{L_{k-1}}, \, \L_{L_0}:=\emptyset. \ee

Now take two different numbers $q_1,q_2\in\N$ and consider the Delone set $D$ defined by (see Fig. \ref{exdel})
\be\label{Del} D:=\left(\bigcup_{k\in \N_e}q_1\Z\cap A_k \right) \cup \left( \bigcup_{k\in \N_o} q_2\Z\cap A_k \right).\ee

 \begin{figure}[ht]
 \begin{center}
  \scalebox{0.4}{\includegraphics{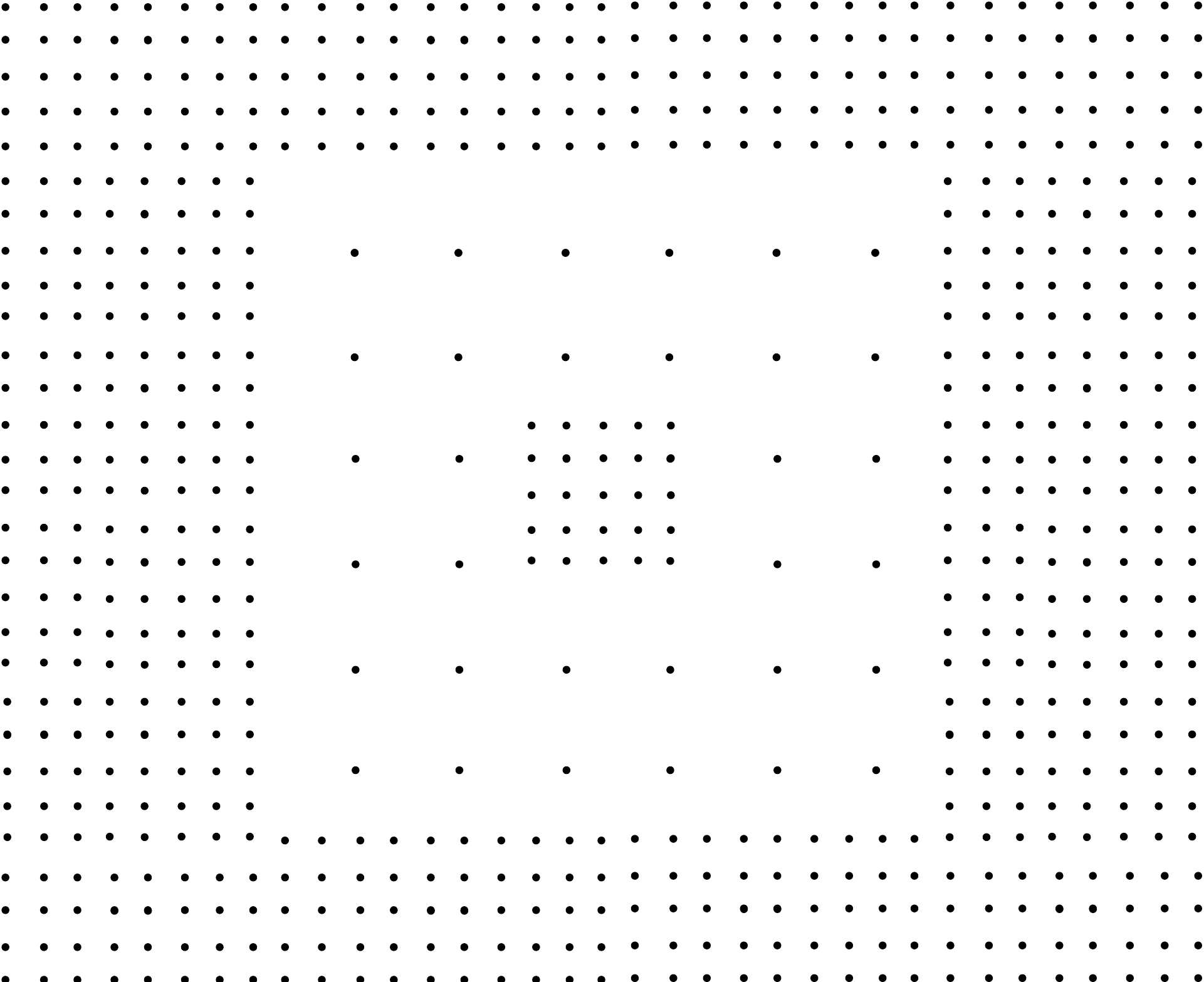}}
 \end{center}
 \caption{The Delone set $D$.}
 \label{exdel}
 \end{figure}

In what follows,  $\abs{A}$ stands for cardinality or Lebesgue measure, depending on $A$ being a discrete set or $A\subset\Rd$.

\begin{prop}\label{nonUE}
The Delone set $D$ defined in \eqref{Del} does not have the \emph{uniform pattern frequency} property given in Definition \ref{def-upf}.  In particular, the Delone dynamical system $X_D$ is not uniquely ergodic.
\end{prop}
\begin{proof}
Note that since $D$ is embedded in $\Z$, it is a set of finite local complexity, in which case the properties of unique ergodicity and uniform pattern frequency are equivalent, see \cite[Proposition 2.32]{MR} (see also \cite[Theorem 1.7]{LS03}, \cite[Theorem 2.7]{LeMo02} which apply to our setting).
Without loss of generality, assume $1\leq q_1<q_2$ and consider the covering of $\R^d$ defined in \eqref{cov}.

Fix  $k\in\N_e$ and take the pattern $Q:=\{ (0,0,...)\}$ consisting of the origin in $\Rd$ with support $B_{1/2}(0)$.  Since $D$ consists of disjoint translations of $Q$, we have that the number of $B_{1/2}(0)$-patterns in $\L_{L_k}$ that are translations of $Q$ by an element of $\L_{L_k}$ is given by
 \begin{align}  \tilde\eta_k(Q) := & \left| \left\{\tilde Q \subset D\,:\,\exists y\in \L_{L_{k}}\, \, \mbox{ such that }\, y+\tilde Q=Q\right\}\right| \nonumber\\
 =& \left| \left\{\tilde Q \subset D\,:\,\exists y\in A_k  \,  \, \mbox{ such that }\, y+\tilde Q=Q \right\} \right| \nonumber\\
 &   \quad + \left|  \left\{\tilde Q \subset D\,:\,\exists y\in \L_{L_{k-1}} \,  \, \mbox{ such that }\,y+\tilde Q=Q  \right\}\right|\nonumber \\
 = & | A_k \cap q_1 \Z | + C_{L_{k-1}},
 \end{align}
 where $C_{L_{k-1}} \approx \abs{\L_{L_{k-1}}\cap q_2\Z}$.

For $k\in\mathbb N_e$, we have
\begin{align}  \frac{ \tilde\eta_{k}(Q)}{\abs{\L_{L_k}}}  = \frac{1}{\abs{\L_{L_k}}} \left( | A_k \cap q_1 \Z | + C_{L_{k-1}}\right). \end{align}
Analogously, we obtain for $k\in\N_o$,
\be  \frac{ \tilde\eta_{k}(Q)}{\abs{\L_{L_k}}} =\frac{1}{\abs{\L_{L_k}}} \left( | A_k \cap q_2 \Z |+ C'_{L_{k-1}}\right), \ee
where $ C'_{L_{k-1}}\approx \abs{\L_{L_{k-1}}\cap q_1\Z} $.
Noting that $| A_k \cap q_i \Z |=q_i^{-d} \abs{A_k} + o(q_i^{-d} \abs{A_k})$ for $i=1,2$, and recalling that $A_k=\overline{\L}_{L_k}\setminus \L_{L_{k-1}}$ and $L_{k}=L_{k-1}^\alpha$ with $\alpha>1$, we see that a subsequence of  $\left(\eta_k(Q)\right)_k$ with $k\in\N_e$ converges to $q_1^{-d}$, while a subsequence with $k\in\N_o$, converges to $q_2^{-d}$.
Therefore, $D$ does not have the property of uniform pattern frequency.
\end{proof}

Now consider the Delone--Anderson Hamiltonian $H_\dom$ associated to $D$, given by

\be\label{app_H} H_\dom= H_0+\displaystyle\sum_{p\in D} \omega_p u(x-p)
\ee
 and assume \Az and \Ao holds. Since $D$ is not uniquely ergodic, Corollary \eqref{eids}(i) gives the existence of the integrated density of states only for $\mu$-a.e. $P\in X_D$ and $\P_P$-a.e. $\omega\in\Omega_P$, where $\mu$ is a (not necessarily unique) ergodic measure on the hull $X_D$.
The $\mu$-a.e. convergence in Corollary \ref{eids} does not hold for $P=D$.

\begin{prop}\label{in_the_hull}
The lattices $D_1:=q_1\Z$ and $D_2:=q_2\Z$ belong to $X_D=\overline{\{D+x\,:\, x\in\Rd\}}$.
\end{prop}

\begin{proof} We will show that there exists a sequence $(P_k)_{k\in\mathcal J}$ in $X_D$, for some index set $\mathcal J$, that converges to $D_1$ in the vague topology. Let us recall that this is equivalent to say that for every compact set $K\subset \Rd$, for every $\epsilon>0$ for finally all $k\in\N$, the following inclusions hold \cite[Lemma 2.8]{MR}
\be\label{vaguec} P_k \cap K\subset (D_1)_\epsilon \quad \mbox{and} \quad D_1\cap K\subset (P_k)_\epsilon,
\ee
where the $\epsilon$-thickened version of a set is defined after Eq. \eqref{thickening}.

Given a compact set $K\in\Rd$, there exists $k_e(K)\in\N_e$ such that for all $k\in\N_e$, $k\geq k_e(K)$, one can find a vector ${\bf x}_k\in\Rd$ such that
\be K\subset {\bf x}_k+A_k \quad \text{and}\quad \left({\bf x}_k+A_k\right)\cap D\subset D_1. \ee
We write $\N_e':=\{k\in\N_e\,:\, k\geq k_e(K)\}$ and define $P^e_k:={\bf x}_k+D$ for $k\in\N_e'$.  Recalling \eqref{Del}, we have that  $P^e_k\cap K=D_1\cap K$, for all $k\in\N_e'$. In particular, \eqref{vaguec} holds, therefore $(P^e_k)_{k\in\N_e'}$ converges to $D_1$ in the vague topology of $X_D$.
%

The same argument applied to $k\in\N_o$ proves that there exists a sequence $(P_k^o)_{k\in\N_o'}\subset X_D$, with $\N_o'\subset\N_o$, that converges to $D_2$ in the vague topology in $X_D$.
\end{proof}

\bigskip

Now, consider the measure $\mu_{D_1}$ defined on $X_D$ by

\be \mu_{D_1}(\mathcal B)=\frac{1}{q_1^d}\abs{\{t\in[0,q_1]^d\,:\, t+D_1 \in \mathcal B \}}, \ee
for any measurable set $\mathcal B\subset X_D$. This is an ergodic measure such that for every $x\in\Rd$, $x+D\notin \supp \mu_{D_1}$. Taking this measure, Corollary \ref{eids}(i) states that the integrated density of states for $H_\pom$ exists for $\mu_{D_1}$-a.e. $P\in X_D$ and $\P_P$-a.e. $\omega\in \Omega_P$.
Therefore, we obtain information on the integrated density of states for a family of periodic sets, translates of $D_1$, but no information relative to the aperiodic set $D$. This is no surprise, considering the following

\begin{prop}\label{noIDS} Let $H_D$ be the Delone operator defined by
\be H_D:=H_0+ \displaystyle\sum_{p\in D} u(x-p),\ee
where $H_0$ is as in \eqref{app_H}. Let $\nu_{L}^D$ be the finite-volume integrated density of states of $H_D$ as in
\eqref{fvids}.  The limit of $\nu_{L}^D$ when $L$ tends to infinity does not exist.
\end{prop}

\begin{proof}
Let $\nu^{D_1}$ and $\nu^{D_2}$ be the integrated density of states of the Delone operators associated to $D_1$ and $D_2$, respectively. Following the reasoning in the proof of Proposition \ref{nonUE} one can show that, for $F\in C_c(\R)$, the measure $\d\nu_{L_k}$ defined through Eq. \eqref{fvids} is obtained by taking the limit when $k\rightarrow\infty$ of the following quantity:
\be \frac{1}{\abs{\L_{L_k}}}\tr \left(F(H)\chi_{\L_{L_k}}\right)  =  \frac{1}{\abs{\L_{L_k}}} \tr \left( F(H) \left(\chi_{A_k} + \chi_{\L_{L_{k-1}}} \right)\right)
\ee
The second term in the r.h.s. is negligible, while the first term tends to $\d\nu^{D_1}$ if one takes a subsequence $L_k$ with $k\in\N_e$, and to $\d\nu^{D_2}$, if $k\in\N_o$.

\end{proof}

\section*{Acknowledgements}

PM received partial financial support from the German Research Council (DfG) through Sfb/Tr 12. CRM received financial support from  the European Community Seventh Framework Programme FP7 under grant agreement number 329458 (ETAM).

\end{document}